\documentclass[manyauthors]{fundam}

\publyear{22}
\papernumber{2102}
\volume{185}
\issue{1}

\usepackage{url} 
\usepackage[ruled,lined]{algorithm2e}
\usepackage{graphicx}
\usepackage{tikz}

\usetikzlibrary{automata, positioning, arrows}

\usepackage{xspace}
\usepackage{hyperref}
\usepackage{mwe}
\usepackage{wrapfig}
\usepackage{paralist}

\usepackage{todonotes}

\newcommand{\nPN}{$\nu$PN\xspace}
\newcommand{\nPNs}{$\nu$PNs\xspace}

\usepackage{listings}
\usepackage{comment}

\todostyle{fdc}{color=green}
\todostyle{fdcin}{color=green,inline}
\todostyle{tp}{color=yellow}
\todostyle{tpin}{color=yellow,inline}

\usepackage{amsfonts}
\usepackage{amssymb}
\usepackage{caption}
\usepackage{subcaption}
\usepackage{multirow}

\usepackage{thm-restate}
\usepackage{thmtools}
\usepackage{empheq}
\usepackage[inline, shortlabels]{enumitem}
\usepackage{cleveref}

\newcommand{\multiset}[1]{\boldsymbol{#1}}

\newcommand{\os}{\ensuremath{\mathfrak{E}}}
\newcommand{\tup}[1]{\langle#1\rangle}

\newcommand{\supp}{\mathit{supp}}

\newcommand{\sqleq}{\sqsubseteq}

\newcommand{\support}[1]{\texttt{Supp}(#1)}

\newcommand{\prefun}{{\ensuremath{\tt pre}}}
\newcommand{\postfun}{{\ensuremath{\tt post}}}

\usepackage{tikz}
\usetikzlibrary{patterns,shapes,arrows,calc,fit,arrows.meta,decorations.pathmorphing,arrows.meta,decorations.pathreplacing}
\usetikzlibrary{arrows,petri,automata,positioning}
\tikzstyle{place}=[circle,minimum height=6mm,draw]
\tikzstyle{transvert}=[rectangle,minimum width=2mm, minimum height=8mm,draw]
\tikzstyle{transhor}=[rectangle,minimum width=8mm, minimum height=2mm,draw]
\tikzset{>={Stealth[scale=1.2]}}
\tikzstyle{empty}=[circle,minimum height=6mm]

\tikzset{Rightarrow/.style={double equal sign distance,>={Implies},->},
myDouble/.style={Rightarrow,double},
dashDouble/.style={Rightarrow,double,dashed},
triple/.style={-,preaction={draw,Rightarrow}},
quadruple/.style={preaction={draw,Rightarrow,shorten >=0pt},shorten >=1pt,-,double,double
distance=0.2pt}}

\newcommand{\C}{\mathcal{C}}
\newcommand{\D}{\mathcal{D}}
\newcommand{\E}{\mathfrak{E}}
\newcommand{\F}{\mathcal{F}}

\newcommand{\K}{\mathcal{K}}
\newcommand{\M}{\mathcal{M}}
\newcommand{\N}{\mathcal{N}}

\newcommand{\R}{\mathcal{R}}
\newcommand{\T}{\mathcal{T}}

\newcommand{\X}{\mathcal{X}}

\newcommand{\fmset}[1]{\{\!\{#1\}\!\}}
\newcommand{\nestTok}{\T}

\renewcommand{\blacksquare}{\blacktriangle}

\newcommand*{\defeq}{\stackrel{\text{def}}{=}}
\newcommand{\var}{\text{Var}}
\newcommand{\out}{out_{\Upsilon}(t)}
\newcommand{\abs}[1]{\lvert#1\rvert}
\newcommand{\GnPN}{c-\nPN}
\newcommand{\GnPNs}{c-\nPNs}
\newcommand{\rnPN}{r-\nPN}

\tikzstyle{triangle}=[draw, regular polygon, regular polygon sides=3]
\tikzstyle{pentagon}=[draw, regular polygon, regular polygon sides=5,minimum height=6mm,]
\tikzset{>={Stealth[scale=1.2]}}

\makeatletter
\def\dasharrowfill@#1#2#3#4{
        $\m@th
        \thickmuskip0mu
        \medmuskip\thickmuskip
        \thinmuskip\thickmuskip
        \relax
        #4#1\mkern2mu
        \xleaders\hbox{$#4\mkern2mu#2\mkern2mu$}\hfill
        \mkern2mu
        #3$
}

\def\dashleftarrowfill@{\dasharrowfill@\leftarrow\relbar\relbar}
\def\dashrightarrowfill@{\dasharrowfill@\relbar\relbar\rightarrow}
\def\dashleftrightarrowfill@{\dasharrowfill@\leftarrow\relbar\rightarrow}
\def\dashLeftarrowfill@{\dasharrowfill@\Leftarrow\Relbar\Relbar}
\def\dashRightarrowfill@{\dasharrowfill@\Relbar\Relbar\Rightarrow}
\def\dashLeftrightarrowfill@{\dasharrowfill@\Leftarrow\Relbar\Rightarrow}

\providecommand*\xdashleftarrow[2][]{%
  \ext@arrow 0055{\dashleftarrowfill@}{#1}{#2}}
\providecommand*\xdashrightarrow[2][]{%
  \ext@arrow 0055{\dashrightarrowfill@}{#1}{#2}}
\providecommand*\xdashleftrightarrow[2][]{%
  \ext@arrow 0055{\dashleftrightarrowfill@}{#1}{#2}}
\providecommand*\xdashLeftarrow[2][]{%
  \ext@arrow 0055{\dashLeftarrowfill@}{#1}{#2}}
\providecommand*\xdashRightarrow[2][]{%
  \ext@arrow 0055{\dashRightarrowfill@}{#1}{#2}}
\providecommand*\xdashLeftrightarrow[2][]{
  \ext@arrow 0055{\dashLeftrightarrowfill@}{#1}{#2}}
\makeatother

\newcommand{\Id}{\mathrm{Id}}

\newcommand{\pactive}{p_{\text{active}}}

\newcommand{\specialtrans}{T_{spl}}

\newcommand{\pinit}[1]{p^{fire}_{#1}}
\newcommand{\pfire}[1]{p^{rename}_{#1}}

\newcommand{\start}{\mathit{start}}
\newcommand{\finish}{\mathit{finish}}
\newcommand{\norm}[1]{|#1|}

\makeatletter
\newenvironment{lemmanum}[1]{
  \medskip
  \noindent
  {\bf Lemma~#1.} 
}{
  \rm
  \medskip
}
\makeatother

\begin{document}

\title{The Complexity of Coverability-Like Problems in Elementary Object Systems: Data-Nets to the Rescue}


\author{Francesco Di Cosmo\corresponding\\
Free University of Bozen-Bolzano, Italy\\
frdicosmo@unibz.it
\and Soumodev Mal\\
Chennai Mathematical Institute, India\\
soumodevmal@cmi.ac.in
\and Tephilla Prince\\
IIT Dharwad, India\\
tephilla.prince.18@iitdh.ac.in
} 

\maketitle

\runninghead{F. Di Cosmo, S. Mal, T. Prince}{The Complexity of Coverability-Like Problems in EOSs}

\begin{abstract}
  Elementary Object Systems (EOSs) are a model in the nets-within-nets (NWNs) paradigm, where tokens in turn can host standard Petri nets. We study the complexity of coverability-like problems, including termination and boundedness, over EOSs. Since coverability and boundedness are undecidable in general on EOSs, we focus on the relevant fragment of conservative EOSs (cEOSs). Our technique interprets cEOSs into the framework of data nets, whose tokens carry data from an infinite domain, thus bridging the nesting and the data-aware paradigms. Specifically, we show that cEOS coverability-like problems are equivalent to the coverability-like problems over an interesting fragment, called channel-\nPNs (\GnPNs), of data nets that extends \nPN (featuring globally fresh name creation) with restricted forms of transfers with renaming. \GnPNs remain less expressive than Unordered Data Nets, which feature lossy name creation as well as powerful forms of whole-place operations and broadcasts.  These reductions allow us to analyze cEOS coverability taking advantage of known results on data nets. We conclude that the complexity of cEOS coverability is double-Ackermanian, {$\F_{\omega 2}$-complete}, while termination and boundedness are non-primitive recursive.
 
\end{abstract}


\section{Introduction}
Recent works have studied the Nets Within Nets (NWN) paradigm~\cite{DBLP:conf/apn/Kohler-Bussmeier23a,OurPNSE24,DBLP:conf/ac/Valk03}, i.e., Petri Nets (PNs) whose tokens in turn carry PNs, as a model for the robustness of multiagent systems against agent breakdowns and, more generally, agent imperfections modeled as token losses. These works focus on the reachability/coverability problems of Elementary Object Systems (EOS), i.e., NWNs where there is only one level of nesting. Other forms of NWNs can be found in~\cite{DBLP:conf/ac/Valk03,kohler-busmeier_survey_2014,DBLP:journals/topnoc/Kohler-BussmeierR23,DBLP:journals/fuin/Lomazova00,DBLP:conf/ershov/LomazovaS99}. Out of the several combinations of problem type (reachability and coverability), lossiness degree (none, finite, unbounded number of token losses), and level (at the outer, nested, or both levels), only reachability/coverability of EOSs under an unbounded amount of lossiness at both levels is decidable~\cite{OurPNSE24}. The picture is moderately more optimistic when the constraint of conservativity is applied. In EOSs, each place can host tokens of a fixed type; in a conservative EOS (cEOS), if a transition consumes an object of a given type, at least one object of the same type must be produced, i.e., the set of types available in the net is conserved. 
When subjected to lossiness at both levels, the cEOS reachability/coverability problem is equivalent to perfect cEOS coverability (see~\cite{OurArxiv2}), which is decidable. Instead, perfect cEOS reachability is known to be undecidable~\cite{kohler-busmeier_survey_2014}. Hence, the decidability boundary of reachability/coverability of lossy/perfect EOS/cEOS is fully charted~\cite{OurPNSE24}.
However, the precise complexity class of lossy EOS reachability and perfect cEOS coverability is unknown.

In this paper, we study the complexity of cEOS coverability-like problems, specifically coverability, termination (whether there is an infinte run from an initial configuration), and boundedness (whether the set of reachable configurations from a fixed initial one is finite). It is well known that these problems on PNs share the same complexity, while on some extension, like reset PN, they do not~\cite{SchnoeAckHard10}. It is, thus, unclear what the picture is on cEOSs.

Instead of directly attacking this problem, for example by exploiting techniques for complexity over Well-Structured Transition Systems~\cite{finkel_well-structured_2001,Rosa-Velardo17,MultiplyRecursive,schmitz2012algorithmic}, we bridge the nesting paradigm with the apparently orthogonal extension of Petri nets with data. In fact, while the complexity of verification problems for NWNs have not been well-studied, many results about the complexity of verification for several data extensions of PNs~\cite{LazicNORW08}, whose tokens carry data from an infinite domain, are available in the literature (see Fig.1 in~\cite{LazicS16}).
For \nPN~\cite{DBLP:journals/tcs/Rosa-VelardoF11}
the coverability problem is double-Ackermannian, \(F_{\omega 2}\)-complete~\cite{LazicS16}. 
For Unordered Data Nets (UDNs), the coverability is hyper-Ackermannian, \(F_{\omega^\omega}\)-complete~\cite{Rosa-Velardo17}. Ordered data nets and ordered data Petri nets have both \(F_{\omega^{\omega^\omega}}\)-complete coverability~\cite{Haddad2012}.

Our main result is that the complexity of cEOS coverability-like problems matches the complexity of the corresponding \nPN problems. To do so, we show that nesting can be naturally captured by whole place operations on top of \nPNs and that, as long as only coverability-like problems are of interest, such whole place operations can be abstracted away, resulting in a polynomial reduction from cEOS to \nPN. In turn, tokens labeled by data from an infinite domain, typical of \nPNs, can be regarded as some form of nested nets moving in a cEOS; this provides the opposite reduction from \nPN to cEOS. This paper is based on \cite{OurArxiv2,OurRP25}, preliminary works where we introduced the ideas behind our reductions for coverability only. Compared to those, this work additionally provides: 
\begin{enumerate}
\item A novel example about how the nesting-related semantic features of the cEOS model can be used to capture real scenarios,
\item A revised, streamlined, and systematic technical development, which also highlights the impact of whole place operations on top of \nPNs,
\item The complete proof of the $\F_{\omega2}$-completeness of cEOS coverability, which did not previously appear, and
\item Addressed the problems of termination and boundedness.
\end{enumerate}

The paper is structured as follows.
Sec.~\ref{sec:prelims} provides preliminaries on problems on Configuration Graphs, PNs, \nPNs, and cEOSs. Sec.~\ref{sec:extend} introduces \GnPNs, a form of \nPN with some whole place operations in the syle of UDNs. 
Sec.~\ref{sec:phaseencoding} formalizes the type of encoding we exploit.
Sec.~\ref{sec:fromNutoCEOS},  Sec.~\ref{sec:fromEOS}, and Sec.~\ref{extendedNuPntoNuPn} provide the encoding, and related problem reductions, respectively from \nPNs to cEOSs, from cEOSs to \GnPNs, and from \GnPNs to \nPNs (see Fig.~\ref{fig:reductions}).
 
Finally, Sec.~\ref{sec:conclusions}
discusses conclusions.
\begin{figure}[t]
    \centering
    \begin{tikzpicture}
        \node (n) at(-2,0) {\nPNs};
        \node (c) at(2,0){cEOSs};
        \node (g) at(0,-1){\GnPNs};
        \draw [->](n)-- node [above]{Th.\ref{thm:npnCEOS}} (c);
        \draw [->](c)-- node [below right]{Th.\ref{thm:ceosGnPN}}(g);
        \draw [->](g)-- node [below left]{Th.\ref{thm:brokenChanneltoNu}}(n);
    \end{tikzpicture}
    \caption{Overview of the reductions between \nPNs, cEOSs, and \GnPNs.}
    \label{fig:reductions}
\end{figure}
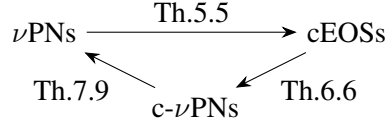

\section{Preliminaries}\label{sec:prelims}
\paragraph*{Configuration Graphs}
\begin{definition}
    A configuration graph (CG) $\C=(C,E)$ is a directed graph where $C$ is a possibly infinite set of configurations and $E$ is a step relation. The usual notion of step sequence (run) applies.
    A \textit{quasi-ordered normed CG} (qonCG) is a tuple $(CG,\leq,\norm{\bullet})$ where $CG=(C,E)$ is a CG, $\leq$ is a quasi-order on $C$, and $\norm{\bullet}:C\longrightarrow \mathbb{N}$ is a function, called norm, which gives the size of the configurations.
\end{definition}
We denote configuration graphs by $\C$, sets of configurations by $C$, and configurations by $K$ or $H$, possibly with indexes, superscripts, etc.
In general, given a CG, we can define the problems of coverability, termination, and boundedness, in an abstract way:
\begin{definition}
    Given a family $F$ of qonCGs, we define the following problems:
    \begin{description}
        \item[$F$-coverability] The inputs are tuples $(\C,K_0,K_1)$ where $\C=(C,E)\in F$ is a CG, $K_0\in C$ is an initial configuration, and $K_1\in C$ is a target configuration; the output is whether there is some configuration $K_2\in C$ such that $K_0\rightarrow^\ast K_2\geq K_1$.
        \item[$F$-termination] The inputs are tuples $(\C,K)$ where $\C=(C,E)\in F$ is a CG, $K\in C$ is an initial configuration; the output is whether there is an infinite run from $K$ in $\C$.
        \item[$F$-boundedness] The inputs are tuples $(\C,K)$ where $\C=(C,E)\in F$ is a CG, $K\in C$ is an initial configuration; the output is whether the set $\{K'\mid K\rightarrow^* K'\}$ of configurations reachable from $\K$ in $\C$ is finite.
    \end{description}
    Given two instances $I_1=(\C_1,K_0^1,\dots,K_n^1)$ of an $F_1$-problem and $I_2=(\C_2,K_0^2,\dots,K_n^2)$ of the corresponding $F_2$-problem, we say that $I_1$ and $I_2$ are \textit{equivalent} if $I_1$ is a yes-instance of the $F_1$-problem iff $I_2$ is a yes-instance of the $F_2$-problem.
\end{definition}

\paragraph*{Multisets.}
Given \(n \in \mathbb{N}\), \([n]=\{1,2,\cdots,n\}\).
A \emph{multiset} $\multiset{m}$ on a set $D$ is a mapping $\multiset{m}:D\rightarrow \mathbb{N}$. The \emph{support} of $\multiset{m}$ is the set $\support{m} = \{i \mid \multiset{m}(i) > 0\}$. The multiset $\multiset{m}$ is finite if $\support{\multiset{m}}$ is finite. The family of all multisets over $D$ is denoted by $D^\oplus$. We denote a finite multiset $\multiset{m}$ by enumerating the elements $d\in\support{\multiset{m}}$ exactly $\multiset{m}(d)$ times in between $\{\{$ and $\}\}$.
The empty multiset $\fmset{}$
is also denoted by $\emptyset$. The empty multiset on the empty domain is denoted by $\varepsilon$.
Given two multisets $\multiset{m_1}$ and $\multiset{m_2}$ on $D$, we define the multisets $\multiset{m_1} + \multiset{m_2}$ and $\multiset{m_1} - \multiset{m_2}$ on $D$ as follows:
$(\multiset{m_1} + \multiset{m_2})(d) = \multiset{m_1}(d)  + \multiset{m_2}(d)$ and
$(\multiset{m_1} - \multiset{m_2})(d) = \max(\multiset{m_1}(d)  - \multiset{m_2}(d),0)$, for each $d\in D$.
We write $\multiset{m_1} \sqleq \multiset{m_2}$ if, for each $d\in D$, we have $\multiset{m_1}(d)  \leq \multiset{m_2}(d)$. The cardinality $\norm{m}$ of a multiset $m$ is $\sum_{d\in\support{m}} m(d)$.

\paragraph*{Petri Nets.}
A PN~\cite{murata89} is a tuple $N=(P,T,F)$, where $P$ is a finite \textit{place set}, $T$ is a finite \textit{transition set}, and $F: (P\times T) \cup (T\times P) \longrightarrow \mathbb{N}$ is a \textit{flow function}. We set the functions of \textit{pre-} and \textit{post-conditions} ${\prefun}_N,\postfun_N : T\rightarrow ( P\rightarrow \mathbb{N})$ where $\prefun_N(t)(p)=F(p,t)$ and $\postfun_N(t)(p)=F(t,p)$. 
A \textit{marking} $\mu$ is a finite multiset on $P$.
A 
$t\in T$ is enabled on $\mu$ if, for each $p\in P$, we have $\prefun_N(t)(p)\leq \mu(p)$. Its firing results in the marking $\mu'$ such that $\mu'(p)=\mu(p)-\prefun_N(t)(p)+\postfun_N(t)(p)$, for each $p\in P$. 
We always assume that $P$ is ordered. Thus, we can denote a marking $m$ as a multiset or as a vector $\tup{\multiset{m}(p)}_{p\in P}$.
We also work with the special \emph{empty PN} $\blacksquare=(\emptyset,\emptyset,\emptyset)$, whose only marking is $\varepsilon$. Places are depicted by circles, transitions by rectangles, markings by multisets of black tokens $\bullet$ in the respective places, and the flow by labeled arrows (see Fig.~\ref{fig:PNInit} and Fig.~\ref{fig:PNPost}).

The qonCG of a $N$ has markings as configurations, transition firings as steps, component-wise $\leq$ as quasi-order, and the token count function as norm.
\begin{definition}[PN qonCG]
    Given a PN $N=(P,T,F)$, the \textit{qonCG of $N$} is the tuple $\C_N=((N^{|P|},E),\sqleq,\norm{\bullet}_N)$ where, for $\mu_1,\mu_2\in N^{|P|}$, $E(\mu_1,\mu_2)$ iff there is some transition firing such that $\mu_1\rightarrow \mu_2$, and $\norm{\mu_1}_N=\norm{\mu_1}$, the cardinality of $\mu$ as a multiset.
\end{definition}

\paragraph*{\texorpdfstring{\nPN}{nuPN}.}
A \nPN is a PN where each token is associated with a data value that comes from a countable domain. The flow function is extended to check equality/inequality of data values using a finite set of variables as labels.
We recall \nPN as in~\cite{LazicS16}. Let $\Upsilon$ and $\X$ be disjoint sets of \textit{fresh} and \textit{standard variables}, denoted by $x_i$ and $\nu_i$ for $i\in\mathbb{N}$, respectively. 
Let $Vars\defeq \X \bigcup \Upsilon$.
\begin{definition}\label{dfn:nuPN}
A \nPN is a PN $\D=\tup{P,T,F}$ with the provision that $F:(P\times T) \bigcup (T \times P) \to Vars^{\oplus}$ and, for each $t\in T$, $\Upsilon\cap\prefun(t)=\emptyset$ and $\postfun(t)\setminus\Upsilon\subseteq\prefun(t)$, where $\prefun(t)=\bigcup_{p\in P} \supp(F(p,t))$ and $\postfun(t)=\bigcup_{p\in P} \supp(F(t,p))$.
\end{definition}
For each $t\in T$, we set $\var(t)=\prefun(t)\cup\postfun(t)$. We require that, for each transition $t$, $\var(t)=\{x_1,\dots,x_{\norm{\var(t)\cap \chi}}\} \cup \{\nu_1,\dots,\nu_{\norm{\var(t)\cap \Upsilon}}\}$. In this section, we work with a fixed arbitrary \nPN $\D=\tup{P,T,F}$ where $P=\{p_1,\dots,p_\ell\}$.
The flow $F_x$ of a variable $x\in \var$ is $F_x:(P\times T) \bigcup (T \times P) \to \mathbb{N}$ where $F_x (p,t)\defeq F(p,t)(x)$ and $F_x(t,p)\defeq F(t,p)(x)$. We denote 
$\tup{F_x(p_1,t),\dots,F_x(p_\ell,t)}\in \mathbb{N}^\ell$ by $F_x(P,t)$ and 
$\tup{F_x(t,p_1),\dots,F_x(t,p_\ell)}\in \mathbb{N}^\ell$ by $F_x(t,P)$.
The set of \textit{configurations} of $\D$ is the set $(\mathbb{N}^P)^\oplus$. For each $t\in T$, let
$
    \out\defeq \sum_{\nu\in \Upsilon(t)} \fmset{F_\nu(t,P)}
$.
Given a configuration $M=\fmset{m_1,\dots,m_{\abs{M}}}$, a transition $t$ is fireable from $M$ if there is a function $e:\X(t)\longrightarrow\{1,\dots,\abs{M}\}$, called \textit{mode}, such that, for each $x\in\X(t)$, $F_x(P,t)\leq m_{e(x)}$. We write $M\rightarrow^{t,e} M'$ if, for some configuration $M''$, 
$M=M''+\sum_{x\in\X(t)}\fmset{m_{e(x)}}
$ and $
M'=M''+\out+\sum_{x\in\X(t)}\fmset{m_{e(x)}'}
$
where, for $x\in\X(t)$, $m'_{e(x)}=m_{e(x)}-F_x(P,t)+F_x(t,P)$. The firing of $t$ with mode $e$ over $M$ applies $F_x$, for each $x\in \X(t)$, to a distinct tuple $m\in M$ such that $m_{e(x)}\geq F_x(P,t)$ and replaces it with $m'_{e(x)}$. It also adds the new markings $F_{\nu}(t,P)$ for $\nu\in\Upsilon$. 
Each tuple $m$ in a configuration $M$ is depicted similarly to PN, but using a dedicated symbol in place of $\bullet$ (see Fig.~\ref{fig:NuPNInit} and Fig.~\ref{fig:NuPNPost}). 

The quasi order among \nPN configurations is a generalization of the order $\sqleq$ among PN markings. Specifically, given two configurations $M_1=\fmset{m_1,\dots,m_n}$ and $M_2=\fmset{k_1,\dots,k_\ell}$, we write $M_1\preceq M_2$
if there is an injective function $e:\{1,\dots,n\}\rightarrow \{1,\dots,\ell\}$ such that $m_i\sqleq k_{e(i)}$ for each $i\in\{1,\dots,n\}$. Then, the qonCG of a \nPN is defined analogously to PNs, where the norm is, again, the token count function.
\begin{definition}[\nPN qonCG]
    Given a \nPN $\D=\tup{P,T,F}$, the \textit{qonCG of $\D$} is the tuple $\C_{\D}=(((\mathbb{N}^P)^\oplus,E),\preceq,\norm{\bullet}_{\D})$ where, for $M_1,M_2\in (\mathbb{N}^P)^\oplus$, 
    \begin{enumerate}
        \item $E(M_1,M_2)$ iff there is some transition firing such that $M_1\rightarrow M_2$,    
    and 
    \item $\norm{M_1}_{\D}=\sum_{m\in\support{M_1}}(M(m)\norm{m})$,
    \end{enumerate}
\end{definition}
\nPN-coverability is is known to be $\F_{\omega2}$-complete\cite{LazicS16}, while \nPN-termination and \nPN-boundedness are known to be non-primitive recursive~\cite{DBLP:journals/tcs/Rosa-VelardoF11}.

\begin{figure}[t]
    \centering
    \begin{subfigure}[b]{.24\textwidth}\centering
    \scalebox{0.7}{
    \begin{tikzpicture}

\node[place,label={[name=p2Lab]left:\scriptsize $p_2$},tokens=1](p2)at (.2,-0.5){};
\node[place,label={[name=p1Lab]left:\scriptsize $p_1$},tokens=4](p1)at (.2,.5){};

\node[transvert] (t)at (1.5,0){};
\node at (t){\scriptsize $t$};

\node[place,label={[name=p3Lab]right:\scriptsize $p_3$},tokens=1](p3)at (2.8,.75){};
\node[place,label={[name=p4Lab]right:\scriptsize $p_4$},tokens=1](p4)at (2.8,0){};
\node[place,label={[name=p5Lab] right:\scriptsize $p_5$},tokens=1](p5)at (2.8,-.75){};

\draw [->] (p1) to node [above, sloped]  (TextNode1) {\scriptsize $2$} (t.north west);

\draw [->] (p2) -- node[above,midway,sloped]{}(t.south west);
\draw [->] (t.north east) -- node[above,midway,sloped] {}(p3);

\draw [->] (t) -- node[above,midway,sloped,pos=.75] {}(p4);
\draw [->,bend right] (t.south east) --node[above,midway,sloped] {\scriptsize 2} (p5);

\end{tikzpicture}
    }
\caption{}
\label{fig:PNInit}        
    \end{subfigure}
\hfill
    \begin{subfigure}[b]{.24\textwidth}\centering
       \scalebox{0.7}{   
   \begin{tikzpicture}

\node[place,label={[name=p2Lab]left:\scriptsize $p_2$}](p2)at (.2,-0.5){};
\node[place,label={[name=p1Lab]left:\scriptsize $p_1$},tokens=2](p1)at (.2,.5){};

\node[transvert] (t)at (1.5,0){};
\node at (t){\scriptsize $t$};

\node[place,label={[name=p3Lab]right:\scriptsize $p_3$},tokens=2](p3)at (2.8,.75){};
\node[place,label={[name=p4Lab]right:\scriptsize $p_4$},tokens=2](p4)at (2.8,0){};
\node[place,label={[name=p5Lab] right:\scriptsize $p_5$},tokens=3](p5)at (2.8,-.75){};

\draw [->] (p1) to node [above, sloped]  (TextNode1) {\scriptsize $2$} (t.north west);

\draw [->] (p2) -- node[above,midway,sloped]{}(t.south west);
\draw [->] (t.north east) -- node[above,midway,sloped] {}(p3);

\draw [->] (t) -- node[above,midway,sloped,pos=.75] {}(p4);
\draw [->,bend right] (t.south east) --node[above,midway,sloped] {\scriptsize 2} (p5);

\end{tikzpicture}
    }     
    \caption{}
\label{fig:PNPost}
    \end{subfigure}
\hfill
    \begin{subfigure}[b]{.24\textwidth}\centering
            \scalebox{0.7}{
    \begin{tikzpicture}

\node[place,label={[name=p2Lab]left:\scriptsize $p_2$}](p2)at (.2,-0.5){};
\node[place,label={[name=p1Lab]left:\scriptsize $p_1$}](p1)at (.2,.5){};

\node[transvert] (t)at (1.5,0){};
\node at (t){\scriptsize $t$};

\node[place,label={[name=p3Lab]right:\scriptsize $p_3$}](p3)at (2.8,.75){};
\node[place,label={[name=p4Lab]right:\scriptsize $p_4$}](p4)at (2.8,0){};
\node[place,label={[name=p5Lab] right:\scriptsize $p_5$}](p5)at (2.8,-.75){};

\draw [->] (p1) to node [above, sloped]  (TextNode1) {\scriptsize $x_1 x_2$} (t.north west);

\draw [->] (p2) -- node[above,midway,sloped] {\scriptsize $x_3$}(t.south west);
\draw [->] (t.north east) -- node[above,midway,sloped] {\scriptsize $x_2$}(p3);

\draw [->] (t) -- node[above,midway,sloped,pos=.75] {\scriptsize $x_3$}(p4);
\draw [->,bend right] (t.south east) --node[above,midway,sloped] {\scriptsize $2\nu_1\nu_2$} (p5);

\node at($(p1)+(0,.2)$){\tiny a a };
\node at($(p1)$){\tiny b b b};
\node at($(p1)-(0,.2)$){\tiny c c};

\node at($(p2)+(0,.2)$){\tiny a a};
\node at($(p2)$){\tiny b b};
\node at($(p2)-(0,.2)$){\tiny c c};

\node at($(p3)$){\tiny a};
\node at($(p4)$){\tiny b};
\node at($(p5)$){\tiny c};
\end{tikzpicture}
    }
    \caption{}
\label{fig:NuPNInit}
    \end{subfigure}
\hfill
    \begin{subfigure}[b]{.24\textwidth}\centering
        \scalebox{0.7}{   
   \begin{tikzpicture}

\node[place,label={[name=p2Lab]left:\scriptsize $p_2$}](p2)at (.2,-0.5){};
\node[place,label={[name=p1Lab]left:\scriptsize $p_1$}](p1)at (.2,.5){};

\node[transvert] (t)at (1.5,0){};
\node at (t){\scriptsize $t$};

\node[place,label={[name=p3Lab]right:\scriptsize $p_3$}](p3)at (2.8,.75){};
\node[place,label={[name=p4Lab]right:\scriptsize $p_4$}](p4)at (2.8,0){};
\node[place,label={[name=p5Lab] right:\scriptsize $p_5$}](p5)at (2.8,-.75){};

\draw [->] (p1) to node [above, sloped]  (TextNode1) {\scriptsize $x_1 x_2$} (t.north west);

\draw [->] (p2) -- node[above,midway,sloped] {\scriptsize $x_3$}(t.south west);
\draw [->] (t.north east) -- node[above,midway,sloped] {\scriptsize $x_2$}(p3);

\draw [->] (t) -- node[above,midway,sloped,pos=.75] {\scriptsize $x_3$}(p4);
\draw [->,bend right] (t.south east) --node[above,midway,sloped] {\scriptsize $2\nu_1\nu_2$} (p5);

\node at($(p1)+(0,.2)$){\tiny a };
\node at($(p1)$){\tiny b b };
\node at($(p1)-(0,.2)$){\tiny c c};

\node at($(p2)+(0,.2)$){\tiny a a};
\node at($(p2)$){\tiny b b};
\node at($(p2)-(0,.2)$){\tiny c };

\node at($(p3)+(0,.1)$){\tiny a};
\node at($(p3)-(0,.1)$){\tiny b};
\node at($(p4)+(0,.1)$){\tiny b};
\node at($(p4)-(0,.1)$){\tiny c};
\node at($(p5)+(0,.2)$){\tiny c};
\node at($(p5)$){\tiny dd};
\node at($(p5)-(0,.2)$){\tiny f};
\end{tikzpicture}
    }     
    \caption{}
\label{fig:NuPNPost}
    \end{subfigure}

    \caption{(\subref{fig:PNInit}) a simple Petri net and (\subref{fig:PNPost}) the resulting marking on firing t. (\subref{fig:NuPNInit}) A $\nu$-net and (\subref{fig:NuPNPost}) the resulting configuration on firing $t$ with mode $e$ instantiating $x_1$ to (the tuple identified by) $a$, $x_2$ to $b$, $x_3$ to $c$, $\nu_1$ to $d$, and $\nu_2$ to $f$.
}
    \label{fig:nupnex}
\end{figure}
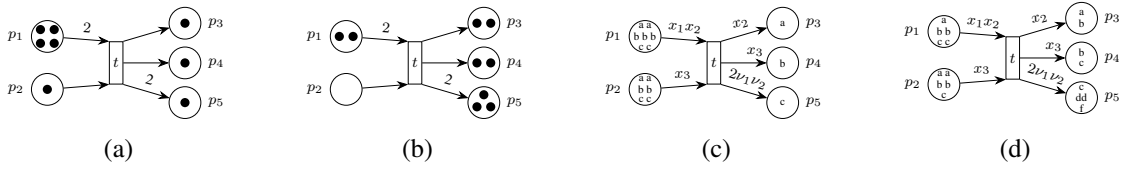
\paragraph*{Elementary Object Systems.}\label{sec:eos}

An EOS~\cite{DBLP:conf/ac/Valk03,kohler-busmeier_survey_2014} is a PN where each token is, in turn, a PN. The underlying PN is known as \emph{system net}, and its tokens, which are PNs, are known as \emph{object net}. A firing can happen asynchronously in either the system net or the object net or synchronously in both.

\begin{definition}\label{def:bussy14_eos}
An \emph{EOS} $\os$ is a tuple $\os=\tup{\hat{N},\N,d,\Theta}$ where:
\begin{enumerate}
\item $\hat{N}=\tup{\hat{P},\hat{T},\hat{F}}$ is a PN called \emph{system net}; $\hat{T}$ contains a special set $ID_{\hat{P}}=\{id_p\mid p\in \hat{P}\}\subseteq \hat{T}$ of \emph{idle transitions} such that, for each distinct $p,q\in \hat{P}$, we have $\hat{F}(p,id_p)=\hat{F}(id_p,p)=1$ and $\hat{F}(q,id_p)=\hat{F}(id_p,q)=0$.
\item $\N$ is a finite set of PNs, called \emph{object PNs}, such that $\blacksquare\in\N$ and if $(P_1,T_1,F_1), (P_2,T_2,F_2)\in\N \cup \{\hat{N}\}$, then $P_1\cap P_2=\emptyset$ and $T_1 \cap T_2 = \emptyset$.
\item $d:\hat{P}\rightarrow \N$ is called the \emph{typing function}. 
\item $\Theta$ is a finite \emph{set of events} where each \emph{event} is a pair $\tup{\hat{\tau},\theta}$, where $\hat{\tau}\in \hat{T}$ and  $\theta:\N \rightarrow \bigcup_{(P,T,F)\in\N} T^\oplus$,
    such that $\theta((P,T,F))\in T^\oplus$ for each $(P,T,F)\in\N$ and, if $\hat{\tau}=id_p$, then $\theta(d(p)) \neq \emptyset$.
\end{enumerate}
\end{definition}

A \textit{nested token} is a system net token carrying an internal marking.
\begin{definition}
Let $\os=\tup{\hat{N},\N,d,\Theta}$ be an EOS. The set of \emph{nested tokens} $\nestTok(\os)$ of $\os$ is the set $\bigcup_{(P,T,F)\in\N} (d^{-1}{(P,T,F)}\times P^{\oplus})$. The set of \emph{nested markings} $\M(\E)$ of $\os$ is $\nestTok(\os)^{\oplus}$.
Given $\lambda,\rho\in \M(\E)$, we say that $\lambda$ is a \emph{sub-marking} of $\mu$ if $\lambda \sqleq \mu$.
\end{definition}

\begin{figure}[t]
    \centering
    \begin{subfigure}[b]{.49\textwidth}\centering
    \scalebox{0.7}{
    \begin{tikzpicture}

\node[place,dashed,label={[name=p2Lab]left:\scriptsize $p_2$},tokens=1](sp2)at (.2,-0.5){};
\node[place,label={[name=p1Lab]left:\scriptsize $p_1$},tokens=2](sp1)at (.2,.5){};
\node[transvert] (st)at (1.5,0){\scriptsize $\hat{t}$};
\node[](lab) at (1.5,-.9){\scriptsize $\{\{ t_1,t_2,t_2\}\}$};
\node[place,label={[name=p3Lab]right:\scriptsize $p_3$}](sp3)at (2.8,.75){};
\node[place,dashed,label={[name=p4Lab]right:\scriptsize $p_4$}](sp4)at (2.8,0){};
\node[place,dotted,label={[name=p5Lab] right:\scriptsize $p_5$}](sp5)at (2.8,-.75){};
\draw [->] (sp1) to node [above, sloped]  (TextNode1) {\scriptsize $2$} (st.north west);
\draw [<->] (sp2) -- node[above,midway,sloped]{}(st.south west);
\draw [->] (st.north east) -- node[above,midway,sloped] {}(sp3);
\draw [->] (st) -- node[above,midway,sloped,pos=.75] {}(sp4);
\draw [->,bend right] (st.south east) --node[above,midway,sloped] {} (sp5);

\begin{scope}[xshift=-3cm,yshift=2cm]
    \node[place,tokens=2,label={below right:\scriptsize $q_1$}](p0)at (0,0){};
    \node[transvert] (t)at (1,0){\scriptsize $t_1$};
    \node[place,label={below left:\scriptsize $q_2$}](p1)at (2,0){};
    \draw[->](p0) -- (t);
    \draw[->](t) -- (p1);
    \node[draw=black,fit={(p0)(t)(p1)},label={right:\scriptsize $N_1$}](obj0){};
    \draw[dashed] (obj0) -- (sp1.center);
\end{scope}

\begin{scope}[xshift=0.5cm,yshift=2cm]
    \node[place,tokens=1,label={below right:\scriptsize $q_1$}](p0)at (0,0){};
    \node[transvert] (t)at (1,0){\scriptsize $t_1$};
    \node[place,label={below left:\scriptsize $q_2$}](p1)at (2,0){};
    \draw[->](p0) -- (t);
    \draw[->](t) -- (p1);
    \node[draw=black,fit={(p0)(t)(p1)},label={right:\scriptsize $N_1$}](obj0){};
\draw[dashed] (obj0) -- (sp1.center);
\end{scope}

\begin{scope}[xshift=-3.5cm,yshift=0cm]
   \node[transvert] (t)at (1,0){\scriptsize $t_2$};
    \node[place,label={below left:\scriptsize $r_1$}](p0)at (2,0){};
    \draw[->](t) -- (p0);
    \node[draw=black,fit={(p0)(t)},label={above:\scriptsize $N_2$}](obj0){};
\draw[dashed] (obj0) -- (sp2.center);
\end{scope}
\end{tikzpicture}
    }
\caption{}
\label{fig:EOSInit}
    \end{subfigure}
    
    \hfill
    
    \begin{subfigure}[b]{.49\textwidth}\centering
       \scalebox{0.7}{   
   \begin{tikzpicture}

\node[place,dashed,label={[name=p2Lab]left:\scriptsize $p_2$},tokens=1](sp2)at (.2,-0.5){};
\node[place,label={[name=p1Lab]left:\scriptsize $p_1$}](sp1)at (.2,.5){};
\node[transvert] (st)at (1.5,0){\scriptsize $\hat{t}$};
\node[](lab) at (1.5,-.9){\scriptsize $\{\{ t_1,t_2,t_2\}\}$};
\node[place,label={[name=p3Lab]right:\scriptsize $p_3$},tokens=1](sp3)at (2.8,.75){};
\node[place,dashed,label={[name=p4Lab]right:\scriptsize $p_4$},tokens=1](sp4)at (2.8,0){};
\node[place,dotted,label={[name=p5Lab] right:\scriptsize $p_5$}](sp5)at (2.8,-.75){};
\node at(sp5){$\blacktriangle$};
\draw [->] (sp1) to node [above, sloped]  (TextNode1) {\scriptsize $2$} (st.north west);
\draw [<->] (sp2) -- node[above,midway,sloped]{}(st.south west);
\draw [->] (st.north east) -- node[above,midway,sloped] {}(sp3);
\draw [->] (st) -- node[above,midway,sloped,pos=.75] {}(sp4);
\draw [->,bend right] (st.south east) --node[above,midway,sloped] {} (sp5);

\begin{scope}[xshift=1.5cm,yshift=2cm]
    \node[place,tokens=2,label={below right:\scriptsize $q_1$}](p0)at (0,0){};
    \node[transvert] (t)at (1,0){\scriptsize $t_1$};
    \node[place,tokens=1,label={below left:\scriptsize $q_2$}](p1)at (2,0){};
    \draw[->](p0) -- (t);
    \draw[->](t) -- (p1);
    \node[draw=black,fit={(p0)(t)(p1)},label={right:\scriptsize $N_1$}](obj0){};
\draw[dashed] (obj0) -- (sp3.center);
\end{scope}

\begin{scope}[xshift=-3.5cm,yshift=0cm]
   \node[transvert] (t)at (1,0){\scriptsize $t_2$};
    \node[place,tokens=1,label={below left:\scriptsize $r_1$}](p0)at (2,0){};
    \draw[->](t) -- (p0);
    \node[draw=black,fit={(p0)(t)},label={above:\scriptsize $N_2$}](obj0){};
\draw[dashed] (obj0) -- (sp2.center);
\end{scope}

\begin{scope}[xshift=3.5cm,yshift=0.5cm]
   \node[transvert] (t)at (1,0){\scriptsize $t_2$};
    \node[place,tokens=1,label={below left:\scriptsize $r_1$}](p0)at (2,0){};
    \draw[->](t) -- (p0);
    \node[draw=black,fit={(p0)(t)},label={above:\scriptsize $N_2$}](obj0){};
\draw[dashed] (obj0) -- (sp4.center);
\end{scope}

\end{tikzpicture}
    }
\caption{}
\label{fig:EOSPost}
    \end{subfigure}
    \caption{An EOS depicting the firing of synchronized event $\langle \hat{t},\fmset{ t_1,t_2,t_2}\rangle$ for a given mode $(\lambda,\rho)$, where $\lambda=\langle p_1,\fmset{q_1,q_1}\rangle+\langle p_1,\fmset{q_1}\rangle+\langle p_2,\fmset{}\rangle$, 
    $\rho=\langle p_2,\fmset{r_1}\rangle+\langle p_3,\fmset{q_1,q_1,q_2}\rangle+\langle p_4,\fmset{r_1}\rangle+\langle p_5,\varepsilon\rangle$, $d(p_1)=d(p_3) = N_1$, 
    $d(p_2)=d(p_4) = N_2$,
    $d(p_5)=\blacktriangle$.}
    \label{fig:eosex}
\end{figure}
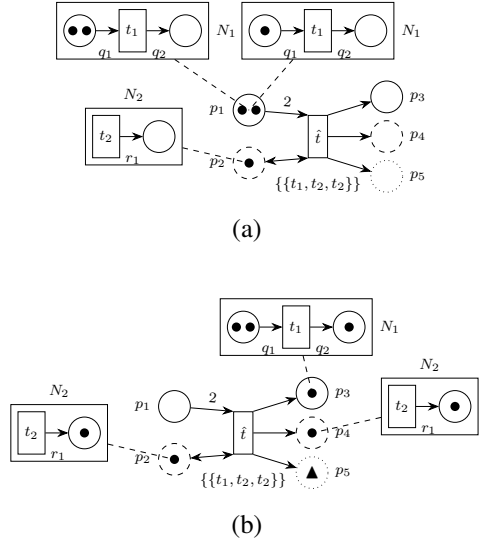

Nested tokens are depicted via a dashed line from the token in the system net place to an instance of the object type where the internal marking is represented in the standard PN way. However, if the nested token is $\tup{p,\varepsilon}$ where $p$ is of type $\blacksquare$, we represent it with a black-token $\blacksquare$ on $p$. 
Events $\tup{\hat{\tau},\theta}$ are depicted by labeling the system net transition $\hat{\tau}$ by the multiset $\theta$ of object net transitions (see Fig.~\ref{fig:eosex}). If there are several events involving $\hat{\tau}$, then $\hat{\tau}$ has several labels.

\begin{example}
    Fig.~\ref{fig:EOSInit} depicts the EOS \(\os = (\hat{N},\N, d, \Theta)\) where we have
    \begin{itemize}
        \item the system net \(\hat{N} = (\{p_1,p_2,p_3,p_4,p_5\},\{\hat{t}\},\hat{F})\) where \(\hat{F}(p_1,\hat{t}) = 2\) and \(\hat{F}(p_2,\hat{t}) = \hat{F}(\hat{t},p_3) = \hat{F}(\hat{t},p_4) = \hat{F}(\hat{t},p_5) = 1\). Otherwise, $\hat{F}$ is the constant function $\varepsilon$.
        \item the set of object nets \(\N = \{N_1,N_2,\blacktriangle\}\), where \(N_1 = (\{q_1,q_2\},\{t_1\},F_1)\), \(N_2 = (\{r_1\},\{t_2\},F_2)\) where \(F_1(q_1,t_1) = F_1(t_1,q_2) = 1\) and \(F_2(t_2,r_1) = 1\).
        \item the typing function \(d\), where \(d(p_1) = d(p_3) = N_1, d(p_2) = d(p_4) = N_2\), and \(d(p_5) = \blacktriangle\).
        \item the set of events \(\Theta = \{\tup{\hat{t},\fmset{t_1,t_2,t_2}}\}\).
    \end{itemize}
Fig.~\ref{fig:EOSInit} also depicts the nested marking $\tup{p_1,\fmset{q_1,q_1}}+\tup{p_1,\fmset{q_1}}+\tup{p_2,\fmset{}}$.
\end{example}

Intuitively, the firing of an event $e=\tup{\tau,\theta}$ can be characterized as follows: first, merge, type by type, all objects handled by the preconditions of $\tau$
(identified by a mode, as in \nPNs), obtaining merged tuples; second, fire, at the same time, all transitions in $\theta$ on the merged tuples; third, non-deterministically distribute the tokens in the updated tuples among new objects in the places indicated by the post-conditions of $\tau$. Possibly, some of the new objects have empty marking. This process is formalized by the next definitions.

\begin{definition}
Let $\os$ be an EOS $\tup{\hat{N},\N,d,\Theta}$. The \emph{projection operator $\Pi^1$} maps each nested marking $\mu=\sum_{i\in I}\tup{\hat{p}_i,M_i}$ for $\E$ to the PN marking $\sum_{i\in I}\hat{p}_i$ for $\hat{N}$. Given an object net $N\in\N$, the \emph{projection operator $\Pi^2_N$} maps each nested marking $\mu=\sum_{i\in I}\tup{\hat{p}_i,M_i}$ for $\E$ to the PN marking $\sum_{j\in J} M_j$ for ${N}$ where $J=\{i\in I\mid d(\hat{p}_i)=N\}$.
\end{definition}

\begin{example}
    Let $M=\tup{p_1,\fmset{q_1,q_1}}+\tup{p_1,\fmset{q_1}}+\tup{p_2,\varepsilon}$ be the nested marking in Fig.~\ref{fig:EOSInit}. We have:
         $\Pi^1(M)=\fmset{p_1,p_1,p_2}$;
         $\Pi^2_{N_1}(M)=\fmset{q_1,q_1}+\fmset{q_1}$;
         $\Pi^2_{N_2}=\fmset{}$;
         $\Pi^2_{\blacksquare}=\varepsilon$.
\end{example}

We now define the \textit{enabledness condition}.
Set $\prefun_{N}(\theta(N))=\sum_{i\in I}\prefun_N(t_i)$ and $\postfun_{N}(\theta(N))=\sum_{i\in I}\postfun_N(t_i)$ where $(t_i)_{i\in I}$ is an enumeration of $\theta(N)$.

\begin{definition}\label{def:bussy14_enable} 
Let $\os$ be an EOS $\tup{\hat{N},\N,d,\Theta}$. Given an event $e=\tup{\hat{\tau},\theta}\in \Theta$ and markings $\lambda,\rho\in\M(\os)$, the \emph{enabledness condition} $\Phi(\tup{\hat{\tau},\theta},\lambda,\rho)$ holds iff
\begin{align*}
\Pi^1(\lambda)=\prefun_{\hat{N}}(\hat{\tau})\ \land \Pi^1(\rho)=\postfun_{\hat{N}}(\hat{\tau})\ \land
\forall N\in \N,\ \Pi^2_N(\lambda)\geq \prefun_N(\theta(N))\ \land\\
\forall N\in\N,\ \Pi^2_N(\rho)=\Pi^2_N(\lambda)-\prefun_N(\theta(N))+\postfun_N(\theta(N))
\end{align*}
The event $e$ is \emph{enabled with mode $(\lambda,\rho)$ on a marking $\mu$} iff $\Phi(e,\lambda,\rho)$ holds and $\lambda\sqleq \mu$.
Its firing results in the step $\mu\xrightarrow{(e,\lambda,\rho)}\mu-\lambda+\rho$.
\end{definition}

\begin{example}
    Let $\lambda=\langle p_1,\fmset{q_1,q_1}\rangle+\langle p_1,\fmset{q_1}\rangle+\langle p_2,\fmset{}\rangle$ and $\rho=\langle p_2,\fmset{r_1}\rangle+\langle p_3,\fmset{q_1,q_1,q_2}\rangle+\langle p_4,\fmset{r_1}\rangle+\langle p_5,\varepsilon\rangle$. In Fig.~\ref{fig:eosex}, the event $e=\tup{\hat{t},\fmset{t_1,t_2,t_2}}$ is enabled under mode $(\lambda,\rho)$. Fig.~\ref{fig:EOSPost} depicts the nested marking reached after firing $e$ under mode $(\lambda,\rho)$. When compared to the intuitive presentation (provided above) of the semantics, we have, for each type $N\in \N$:
    \begin{itemize}
        \item $\Pi^1(\lambda)$ amounts to the consumed system net tokens.
        \item $\Pi^2_{N}(\lambda)$
        is the resultant marking obtained by merging the consumed system net tokens of type \(N\).
        \item $\Pi^2_{N}(\rho)$ amounts to the firing, on the merged marking above, of the transitions in $\fmset{t_1,t_2,t_2}$ from net $N$.
        \item $\Pi^1(\rho)$ amounts to the produced system net tokens.
        \item $\rho-\lambda$ amounts to the distributed marking.
    \end{itemize}
    Note that the non-determinism of the final distribution is captured by the choice of one of the many possible enabling modes.
\end{example}

An EOS is \emph{conservative} if, for each system net transition $t$, if $t$ consumes a nested token on a place of type $N$, then it produces at least one token on a place of that same type. 

\begin{definition}
A cEOS is an EOS $\os=\tup{\hat{N},\N,d,\Theta}$ with $\hat{N} = \tup{\hat{P},\hat{T},\hat{F}}$ where, for all $\hat{t} \in \hat{T}$, $d(\support{\prefun_{\hat{N}}(\hat{t})})\subseteq d(\support{\postfun_{\hat{N}}(\hat{t})})$.
\end{definition}

Let $\leq_f$ be the order among configurations such that $\mu\leq_f\mu'$ if $\mu'$ is obtained from $\mu$ by adding
\begin{enumerate*}
    \item tokens in the inner markings of objects in $\mu$ and/or
    \item nested tokens 
    at the system net places.
\end{enumerate*}
The \textit{EOS coverability problem}
asks, given an EOS $\os$ and configurations $\mu_f$ and $\mu_1$, whether there is a \textit{run} (sequence of event firings) from $\mu_0$ to a configuration $\mu_1\geq_f \mu_f$.
We may denote $\leq_f$ also by $\leq$.

The qonCG of a cEOS is defined analogously to PNs, where the norm is still, the token count function.
\begin{definition}[cEOS qonCG]
    Given a cEOS $\os=\tup{\hat{N},\N,d,\Theta}$, the \textit{qonCG of $\os$} is the tuple $\C_\os=((\M(\os),E),\leq_f,\norm{\bullet}_\os)$ where, for $M_1,M_2\in (\mathbb{N}^P)^\oplus$, 
    \begin{enumerate}
        \item $E(M_1,M_2)$ iff there is some transition firing such that $M_1\rightarrow M_2$, and 
    \item $\norm{M_1}_\os=
    \norm{\Pi^1(M_1)}+\sum_{N\in\N}\norm{\Pi^2_N(M_1)}$.
    \end{enumerate}
\end{definition}

EOS-coverability (respectively, cEOS-coverabilty) is known to be undecidable (decidable; Th. 4.3 and Th. 5.2 in~\cite{kohler-busmeier_survey_2014}).

While the purpose of this paper is not to motivate cEOS, which were introduced in previous works, we illustrate cEOS semantics via the following example.
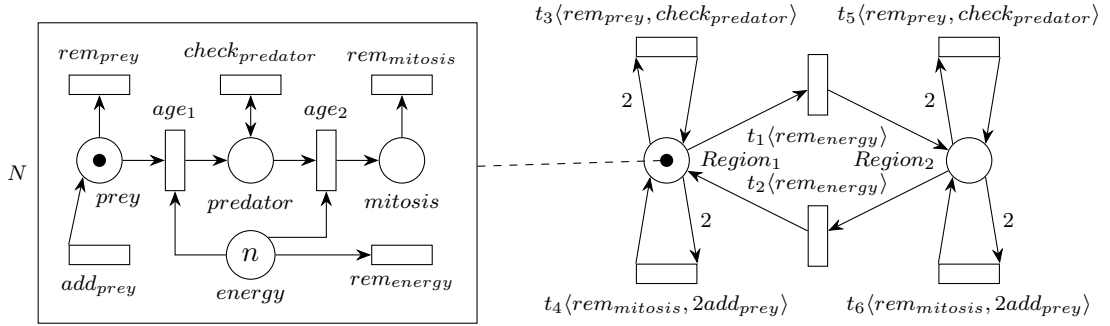
\begin{figure}[t]
\centering
\begin{tikzpicture}
\node[place,tokens=1,label=right:\scriptsize $Region_1$] (reg1) at (0,0){};
\node[transvert,label={[name=t2Lab]below:\scriptsize $t_1\langle rem_{energy}\rangle$}] (t3)at (2,1){};
\node[transvert,label={[name=t2Lab]above:\scriptsize $t_2\langle rem_{energy}\rangle$}] (t4)at (2,-1){};
\node[place,label=left:\scriptsize $Region_2$] (reg2) at (4,0){};
    \draw[->](reg1) -- (t3);
    \draw[->](reg2) -- (t4);
    \draw[<-](reg2) -- (t3);
    \draw[<-](reg1) -- (t4);

    \node[transhor,label={[name=t3Lab]above:\scriptsize $t_3\langle rem_{prey},check_{predator}\rangle$}] (t5)at (0,1.5){};
    \node[transhor,label={[name=t3Lab]below:\scriptsize $t_4\langle rem_{mitosis},2 add_{prey}\rangle$}] (t6)at (0,-1.5){};
    \draw[->](reg1.north west) --node[midway,left]{\scriptsize$2$} (t5.south west);
    \draw[<-](reg1.north east) -- (t5.south east);
    \draw[<-](reg1.south west) -- (t6.north west);
    \draw[->](reg1.south east) --node[midway,right]{\scriptsize$2$} (t6.north east);

    \node[transhor,label={[name=t3Lab]above:\scriptsize $t_5\langle rem_{prey},check_{predator}\rangle$}] (t7)at (4,1.5){};
    \node[transhor,label={[name=t3Lab]below:\scriptsize $t_6\langle rem_{mitosis},2 add_{prey}\rangle$}] (t8)at (4,-1.5){};
    \draw[->](reg2.north west) --node[midway,left]{\scriptsize$2$} (t7.south west);
    \draw[<-](reg2.north east) -- (t7.south east);
    \draw[<-](reg2.south west) -- (t8.north west);
    \draw[->](reg2.south east) --node[midway,right]{\scriptsize$2$} (t8.north east);

\begin{scope}[xshift=-7.5cm,yshift=0cm]
    \node[place,tokens=1,label={[name=p0Lab,xshift=.25cm]below :\scriptsize $prey$}](p0)at (0,0){};
    \node[transvert,label={[name=t1Lab]above:\scriptsize $age_1$}] (t1)at (1,0){};
    \node[place,label={[name=p1Lab]below:\scriptsize $predator$}](p1) at (2,0){};
    \node[transvert,label={[name=t2Lab]above:\scriptsize $age_2$}] (t2)at (3,0){};
    \node[place,label={[name=p2Lab]below:\scriptsize $mitosis$}](p2) at (4,0){};
    \draw[->](p0) -- (t1);
    \draw[->](t1) -- (p1);
    \draw[->](p1) -- (t2);
    \draw[->](t2) -- (p2);

    \node[place,label={[name=resLab]below:\scriptsize $energy$}](res) at (2,-1.25){$n$};
    \node[transhor,label={[name=t3Lab]below:\scriptsize $rem_{energy}$}] (t3)at (4,-1.25){};
    \draw[->](res) -| (t1);
    \draw[->](res.north east) -| (t2);
    \draw[->](res) -- (t3);

  \node[transhor,label={[name=remPreyLab]above:\scriptsize $rem_{prey}$}] (remPrey)at (0,1){};
  \node[transhor,label={[name=addPreyLab]below:\scriptsize $add_{prey}$}] (addPrey)at (0,-1.25){};
  \node[transhor,label={[name=checkLab]above:\scriptsize $check_{predator}$}] (check)at (2,1){};
  \node[transhor,label={[name=remRepLab]above:\scriptsize $rem_{mitosis}$}] (remRep)at (4,1){};

    \draw[->](p0) -- (remPrey);
    \draw[<-](p0.south west) -- (addPrey.north west);
    \draw[<->](p1) -- (check);
    \draw[->](p2) -- (remRep);
    
    \node[draw=black,fit={(p0)(p2)(p0Lab)(p2Lab)(resLab)(remPreyLab)(addPreyLab)(checkLab)(remRepLab)},label={left:\scriptsize $N$}](obj0){};
    \draw[dashed] (obj0) -- (reg1.center);
\end{scope}

\end{tikzpicture}
    \caption{The cEOS of Ex.~\ref{ex:pop}. On top of the depicted events, the set of events $\Theta$ contains also the object autonomous events $\tup{Id_{Region_i},\fmset{age_j}}$ for $i,j\in \{1,2\}$.}
    \label{fig:eosexpop}
\end{figure}
\begin{example}\label{ex:pop}
The cEOS in Fig.~\ref{fig:eosexpop} models a population of unicellular organisms growing, migrating, predating and reproducing. Each organism is represented by an object of type $N$, which comprises a state machine to represent the organism stage (among \textit{prey}, \textit{predator}, and \textit{mitosis}) as well as an extra place \textit{energy}, whose tokens represent available energy. During aging, the organism changes stage by spending some energy. New energy is drained from preys during predation. At mitosis, the energy is non-deterministically split among the two new organisms (possibly, one of the organisms gets zero energy).
The system net represents the environment and features two places $Region_1$ and $Region_2$. Organisms can interact only if they are placed in the same region. Movement among the two regions can happen anytime by spending energy. 
Initially, there is only a single organism, with plenty of energy, in the first region. This is captured by the initial configuration $\tup{Region_1, \fmset{\textit{energy}^n}}$, for some $n\in\mathbb{N}$. Aging is captured by the object autonomous events $
\tup{ id_{\textit{Region}_i}, \textit{age}_j }$ for $i,j\in\{1,2\}$. 

Predation is captured by the event $\tup{t_i, \fmset{rem_{prey}, check_{predator}}}$, for $i\in\{1,2\}$, that consumes two organisms, checks whether one is in the prey stage (removing the corresponding token at the same time) and one is in the predator stage; the energy drain is captured by the the fact that a single object is created, thus by collecting all tokens in both \textit{energy} places into a single one. As long as the organisms are in a legal state, that is there is exactly one token among \textit{prey}, \textit{predator}, and \textit{mitosis}, the firing of this event results in a single organism with a single token on \textit{predator}. 

A careless designer may attempt to capture mitosis in a similar way using a distribution instead of a merging, as done in Fig.~\ref{fig:eosexpop}. Specifically, it might use an event $\tup{ t_i, \fmset{rem_{mitosis},2 add_{prey}}}$, for $i\in\{4,6\}$, which consumes one organism, checks if it is in mitosis phase, and produces two objects, with a total of two tokens on the place \textit{prey}. Unfortunately, the fully non-deterministic behavior of the distribution may result in abominations, i.e., one organism with no token or two tokens on \textit{prey}. This bug is detected by running a coverability problem with target $\tup{
Region_i, \fmset{prey^2}}$, for $i\in\{1,2\}$, which returns \textit{True}.

The bug can be fixed by adding an additional system net place that provisionally hosts the product of the distribution transition and by immediately firing twice an auxiliary transition that checks the presence of one token in \textit{prey} and puts the object into $Region_i$. To fire the transitions in the right sequence, auxiliary places and black tokens are necessary in the system net. If no abomination was created, both firings can take place and the simulation continues. Otherwise, the second firing cannot happen, because one of the abomination has no token in \textit{prey}, and the net reaches a deadlock. However, as long as there is no organism under check by the auxiliary transition, there is no abomination. This is is verified by running a coverability problem with target $\tup{Region_i, \fmset{2 prey}} + \tup{simulating,\varepsilon}$, which returns \textit{False}.

By slightly changing the model, the model designer can check, for a given initial energy $n$, whether the population can prosper up to $m$ specimen, for some $m\in\mathbb{N}$, e.g., by adding an additional system net gadget that collects the tokens from each region. More sophisticated interactions with the environment can be flexibly modeled by, for example, employing additional object types, that is, by taking further advantage of the nested paradigm.
\end{example}

\section{Channel \texorpdfstring{\nPN}{nuPN}}\label{sec:extend}
We introduce \GnPN, an extension of \nPN by restricted forms of whole-place operations in the style of UDNs. These whole-place operations are useful to naturally encode the nested-paradigm into the data-aware paradigm.
\newcommand{\confs}{\texttt{Confs}}
Specifically, \GnPN extend \nPN by \textit{special} transitions that perform transfers with renaming.
The firing of special transitions is similar to UDNs (and affine nets~\cite{FinkelMP04}):
\begin{enumerate*}
    \item 
    the pre-conditions are fired,
    \item 
    the transfers (defined by matrices \(G_t(x_i,x_j)\) for $x_i,x_j\in\X(t)$) are performed
    \item 
    the post-conditions are fired.
\end{enumerate*}

Technically, $G_t(x_i,x_j)\in \{0,1\}^{{P} \times {P}}$. 
$G_t(x_i,x_j)(p,q)=n$ means that, after firing the pre-conditions, 
the transition puts on $q$, in the tuple instantiating $x_j$, $n$ times the number of tokens on $p$ from the tuple instantiating $x_i$.
We use the following notation: if, for some set $P$ of places, $M\in\{0,1\}^{P\times P}$ is a matrix and $p\in P$ is a place, $M[p]$ denotes the row of $M$ indexed by $p$.
Moreover, given a place $p\in P$, we denote by $\delta_p\in\{0,1\}^{\abs{P}}$ the vector such that, for each $q\in P$, $\delta_p(q) = 1$ if and only if $p=q$. For example, given a tuple $m$ and a place $p$, the tuple $m(p)\delta_p$ is the projection of $m$ to $p$.

\begin{definition} \label{def:CNUPN}
    A \emph{channel \nPN (\GnPN)} is a tuple \(N = (P,T,F,G)\) where
         \((P,T,F)\) is a \nPN and
         G maps each \(t \in T\) to a function \(G_t: \X(t) \times \X(t) \rightarrow \{0,1\}^{{P} \times {P}}\) such that, for each variable \(x_i \in \X(t)\) and place $p\in P$:
        \begin{enumerate} [topsep=0pt,itemsep=-1ex,partopsep=1ex,parsep=1ex]
            \item for each variable $x_j\in\X(t)$ either $G_t(x_i,x_j)[p]=0$ or $\exists q\in P\;G_t(x_i,x_j)=\delta_q$, i.e., $G_t(x_i,x_j)[p]$ contains at most a single $1$.
            \item if \(\exists q \in P\,\exists x_j\in\X(t)\setminus\{x_i\}\;G_t(x_i,x_j)[p] = \delta_q\), then \(\forall x_k \in \X(t)\setminus \{x_j\}\; G_t(x_i,x_k)[p] = 0^{\abs{P}}\), else \(G_t(x_i,x_i)[p] = \delta_p\).
        \end{enumerate}
\end{definition}

The graphical representation of a \GnPN $N'=(P,T,F,G)$ adds the representation of $G_t$, for each $t\in T$, on top of the representation of the \nPN $N=(P,T,F)$. This is achieved by special arrows called \textit{channels}.
A channel $c$ is a pair of decorated arrows (e.g., double or triple arrows in Fig.~\ref{fig:CNuPN_example}), one to $t$, called pre-channel of $c$, and one from $t$, called post-channel of $c$. Different channels use different decorations. The pre- and post-channels of $c$ are labeled by two non-empty sequences $\sigma_\text{pre}$ and $\sigma_\text{post}$ of variables, respectively, such that $\abs{\sigma_\text{pre}} =\abs{ \sigma_\text{post}}$. Moreover, for each place $p$, each variable $x\in\X(t)$ can occur at most once in the label of at most one pre-channel from $p$.\footnote{This condition is necessary to forbid the representation of transfers with duplication.}
Intuitively, $c$ represents $\abs{\sigma_\text{pre}}$ transfers from $p$ to $q$ while renaming the tokens taken from the tuple instantiating $\sigma_\text{pre}(i)$ as tokens of the tuple instantiating $\sigma_\text{post}(i)$.
A transition is \textit{special} if it has at least one channel.

    Each PN \(N=(P,T,F)\) corresponds (cf. \GnPN semantics in Def.~\ref{def:stepCNUPN} below) to the \GnPN \(N'=(P,T,F,G)\) where, for each $t\in T$ and $x\in\X(t)$, $y\in \X(t)\setminus\{y\}$, $
    G_t(x,x)=\textbf{Id}$,
    and $G_t(x,y)=0$. Both $N$ and $N'$ have the same graphical representation. $N'$ has no special transition.

\begin{example} \label{example:CNuPN1}
    Fig.~\ref{fig:CNuPN_example} depicts a \GnPN \(N = (P,T,F,G)\) with several configurations (cf. Def.~\ref{def:stepCNUPN} below) where \(P = \{p_1,\cdots,p_5\}\), \(T = \{t\}\), \(\X(t) = \{x_1,x_2,x_3\}\), \(\Upsilon(t) = \{\nu_1,\nu_2\}\), and for each place $p\in P$ and $x\in\var(t)$, $G_t(x_1,x)[p]$ is $\delta_p$ if $x=x_1$ and $0$ otherwise and:
    {\footnotesize \begin{align*}
        F_x(p,t)=&\begin{cases}
        1& \text{if }p=p_2,x\in\{x_2,x_3\}\\
            1& \text{if }p=p_1,x=x_1\\
            0& \text{otherwise}
        \end{cases}
        &
        F_x(t,p)=&\begin{cases}
            1& \text{if }p=p_3,x=x_3\\
            2& \text{if }p=p_3,x=x_2\\
            2& \text{if }p=p_5,x\in\{\nu_1,\nu_2\}\\
            0& \text{otherwise}
        \end{cases}
    \\
    G_t(x_2,x)[p]=&\begin{cases}
        \delta_{p_3}&\text{if }p=p_1, x=x_3\\
        \delta_{p}&\text{if }p\neq p_1, x=x_2\\
        0&\text{otherwise}
    \end{cases}&
    G_t(x_3,x)[p]=&\begin{cases}
        \delta_{p_4}&\text{if }p=p_2, x=x_1\\
        \delta_{p}&\text{if }p\neq p_2, x=x_3\\
        0&\text{otherwise}
    \end{cases}
\end{align*}}%

\end{example}

\begin{figure}[t]
\centering
    \begin{subfigure}{0.24\linewidth}
    \scalebox{0.9}{
        \begin{tikzpicture}

\node[place,label={[name=p2Lab]below:\scriptsize $p_2$}](p2)at (.2,0){};
\node[place,label={[name=p1Lab]above:\scriptsize $p_1$}](p1)at (.2,.75){};

\node[transvert] (t)at (1.5,0){};
\node at (t){\scriptsize $t$};

\node[place,label={[name=p3Lab]above:\scriptsize $p_3$}](p3)at (2.8,.75){};
\node[place,label={[name=p4Lab]below:\scriptsize $p_4$}](p4)at (2.8,0){};
\node[place,label={[name=p5Lab,yshift=.2cm] below right:\scriptsize $p_5$}](p5)at (2.35,-.75){};

\draw [->] (p1) to [bend left=30]  node [above, sloped]  (TextNode1) {\scriptsize $x_1$} (t.north west);
\draw [->] (p2) to [bend right=30]  node [below, sloped]  (TextNode1) {\scriptsize $x_2 x_3$} (t.south west);
        
\draw [->,myDouble] (p1) -- node[above,midway,sloped] {\scriptsize $x_2$}(t);
\draw [->,triple] (p2) -- node[below,midway,sloped] {\scriptsize $x_3$}(t);
\draw [->,myDouble] (t) -- node[above,midway,sloped] {\scriptsize $x_3$}(p3);
\draw [->] (t.north east) to [bend left=30]  node [above, sloped]  (TextNode2) {\scriptsize $2x_2,x_3$} (p3);

\draw [->,triple] (t) -- node[below,midway,sloped,pos=.75] {\scriptsize $x_1$}(p4);
\draw [->,bend right] (t.south east) --node[below,midway,sloped] {\scriptsize $2\nu_1\nu_2$} (p5);

\node at($(p1)+(0,.2)$){\tiny a a };
\node at($(p1)$){\tiny b b b};
\node at($(p1)-(0,.2)$){\tiny c c};

\node at($(p2)+(0,.2)$){\tiny a a};
\node at($(p2)$){\tiny b b};
\node at($(p2)-(0,.2)$){\tiny c c};

\node at($(p3)$){\tiny a};
\node at($(p4)$){\tiny b};
\node at($(p5)$){\tiny c};
\end{tikzpicture}}
        \caption{}
        \label{fig:CNuPN_init}
    \end{subfigure}
    \begin{subfigure}{0.24\linewidth}\scalebox{0.9}{
        \begin{tikzpicture}

\node[place,label={[name=p2Lab]below:\scriptsize $p_2$}](p2)at (0.2,0){};
\node[place,label={[name=p1Lab]above:\scriptsize $p_1$}](p1)at (0.2,.75){};

\node[transvert] (t)at (1.5,0){};
\node at (t){\scriptsize $t$};

\node[place,label={[name=p3Lab]above:\scriptsize $p_3$}](p3)at (2.8,.75){};
\node[place,label={[name=p4Lab]below:\scriptsize $p_4$}](p4)at (2.8,0){};
\node[place,label={[name=p5Lab,yshift=.2cm] below right:\scriptsize $p_5$}](p5)at (2.35,-.75){};

\draw [->] (p1) to [bend left=30]  node [above, sloped]  (TextNode1) {\scriptsize $x_1$} (t.north west);
\draw [->] (p2) to [bend right=30]  node [below, sloped]  (TextNode1) {\scriptsize $x_2 x_3$} (t.south west);
        
\draw [->,myDouble] (p1) -- node[above,midway,sloped] {\scriptsize $x_2$}(t);
\draw [->,triple] (p2) -- node[below,midway,sloped] {\scriptsize $x_3$}(t);
\draw [->,myDouble] (t) -- node[above,midway,sloped] {\scriptsize $x_3$}(p3);
\draw [->] (t.north east) to [bend left=30]  node [above, sloped]  (TextNode2) {\scriptsize $2x_2,x_3$} (p3);

\draw [->,triple] (t) -- node[below,midway,sloped,pos=.75] {\scriptsize $x_1$}(p4);
\draw [->] (t.south east) --node[below,midway,sloped,bend right] {\scriptsize $2\nu_1\nu_2$} (p5);

\node at($(p1)+(0,.2)$){\tiny a};
\node at($(p1)$){\tiny b b b};
\node at($(p1)-(0,.2)$){\tiny c c};

\node at($(p2)+(0,.2)$){\tiny a a};
\node at($(p2)$){\tiny b};
\node at($(p2)-(0,.2)$){\tiny c};

\node at($(p3)$){\tiny a};
\node at($(p4)$){\tiny b};
\node at($(p5)$){\tiny c};
\end{tikzpicture}}
        \caption{}
        \label{fig:CNuPN_step1}
    \end{subfigure}

    \begin{subfigure}{0.24\linewidth}\scalebox{0.9}{
        \begin{tikzpicture}

\node[place,label={[name=p2Lab]below:\scriptsize $p_2$}](p2)at (0.2,0){};
\node[place,label={[name=p1Lab]above:\scriptsize $p_1$}](p1)at (0.2,0.75){};

\node[transvert] (t)at (1.5,0){};
\node at (t){\scriptsize $t$};

\node[place,label={[name=p3Lab]above:\scriptsize $p_3$}](p3)at (2.8,0.75){};
\node[place,label={[name=p4Lab]below:\scriptsize $p_4$}](p4)at (2.8,0){};
\node[place,label={[name=p5Lab,yshift=.2cm] below right:\scriptsize $p_5$}](p5)at (2.35,-0.75){};

\draw [->] (p1) to [bend left=30]  node [above, sloped]  (TextNode1) {\scriptsize $x_1$} (t.north west);
\draw [->] (p2) to [bend right=30]  node [below, sloped]  (TextNode1) {\scriptsize $x_2 x_3$} (t.south west);
        
\draw [->,myDouble] (p1) -- node[above,midway,sloped] {\scriptsize $x_2$}(t);
\draw [->,triple] (p2) -- node[below,midway,sloped] {\scriptsize $x_3$}(t);
\draw [->,myDouble] (t) -- node[above,sloped,midway] {\scriptsize $x_3$}(p3);
\draw [->] (t.north east) to [bend left=30]  node [above, sloped]  (TextNode2) {\scriptsize $2x_2,x_3$} (p3);

\draw [->,triple] (t) -- node[below,midway,sloped,pos=.75] {\scriptsize $x_1$}(p4);
\draw [->] (t.south east) --node[below,midway,sloped] {\scriptsize $2\nu_1\nu_2$} (p5);

\node at($(p1)+(0,.1)$){\tiny a};
\node at($(p1)-(0,.1)$){\tiny c c};

\node at($(p2)+(0,.1)$){\tiny a a};
\node at($(p2)-(0,.1)$){\tiny b};

\node at($(p3)+(0,.2)$){\tiny a};
\node at($(p3)$){\tiny c c c };

\node at($(p4)+(0,.1)$){\tiny a};
\node at($(p4)-(0,.1)$){\tiny b};

\node at($(p5)$){\tiny c};
\end{tikzpicture}}
        \caption{}
        \label{fig:CNuPN_step2}
    \end{subfigure}

    \begin{subfigure}{0.24\linewidth}\scalebox{0.9}{
        \begin{tikzpicture}

\node[place,label={[name=p2Lab]below:\scriptsize $p_2$}](p2)at (0.2,0){};
\node[place,label={[name=p1Lab]above:\scriptsize $p_1$}](p1)at (0.2,0.75){};

\node[transvert] (t)at (1.5,0){};
\node at (t){\scriptsize $t$};

\node[place,label={[name=p3Lab]above:\scriptsize $p_3$}](p3)at (2.8,0.75){};
\node[place,label={[name=p4Lab]below:\scriptsize $p_4$}](p4)at (2.8,0){};
\node[place,label={[name=p5Lab,yshift=.2cm] below right:\scriptsize $p_5$}](p5)at (2.35,-0.75){};

\draw [->] (p1) to [bend left=30]  node [above, sloped]  (TextNode1) {\scriptsize $x_1$} (t.north west);
\draw [->] (p2) to [bend right=30]  node [below, sloped]  (TextNode1) {\scriptsize $x_2 x_3$} (t.south west);
        
\draw [->,myDouble] (p1) -- node[above,midway,sloped] {\scriptsize $x_2$}(t);
\draw [->,triple] (p2) -- node[below,midway,sloped] {\scriptsize $x_3$}(t);
\draw [->,myDouble] (t) -- node[above,sloped,midway] {\scriptsize $x_3$}(p3);
\draw [->] (t.north east) to [bend left=30]  node [above, sloped]  (TextNode2) {\scriptsize $2x_2,x_3$} (p3);

\draw [->,triple] (t) -- node[below,midway,sloped,pos=.75] {\scriptsize $x_1$}(p4);
\draw [->] (t.south east) --node[below,midway,sloped, bend right] {\scriptsize $2\nu_1\nu_2$} (p5);

\node at($(p1)+(0,.1)$){\tiny a};
\node at($(p1)-(0,.1)$){\tiny c c};

\node at($(p2)+(0,.1)$){\tiny a a};
\node at($(p2)-(0,.1)$){\tiny b};

\node at($(p3)+(0,.15)$){\tiny a c};
\node at($(p3)$){\tiny c c c};
\node at($(p3)-(0,.15)$){\tiny b b };

\node at($(p4)+(0,.1)$){\tiny a};
\node at($(p4)-(0,.1)$){\tiny b};

\node at($(p5)+(0,.2)$){\tiny c};
\node at($(p5)$){\tiny d d};
\node at($(p5)-(0,.2)$){\tiny f};

\end{tikzpicture}}
        \caption{}
        \label{fig:CNuPN_step3}
    \end{subfigure}
    \caption{The \GnPN discussed in Ex.~\ref{example:CNuPN1} and Ex.~\ref{example:CNuPN2}, with configurations.
    } 
    \label{fig:CNuPN_example}
\end{figure}
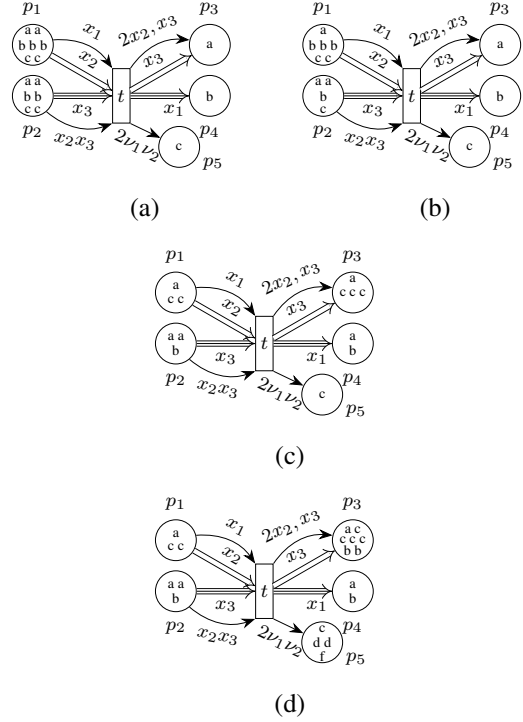

The firing of a special transition $t$ serializes its standard pre-conditions, the transfers, and, finally, the post-conditions. In Def.~\ref{def:stepCNUPN} below, $M''$ captures the tokens uninvolved with the firing, $\out$ accounts for the creation of fresh tuples via $\nu$ variables, $m''$ accounts for the firing of the standard preconditions,
the summation in $m'_{e(x_j)}$
formalizes the transfers to the tuple instantiating $x_j$,
while $m'$ takes into account also the final firing of the standard post-conditions.

\begin{definition} \label{def:stepCNUPN}
A configuration $M$ of a \GnPN $N=(P,T,F,G)$ is a configuration of $(P,T,F)$, and a transition $t\in T$ is enabled in $N$ if it is enabled in $(P,T,F)$. A configuration $M'$ is reached from $M$ by firing a transition $t\in T$ with mode $e$, denoted by $M\rightarrow^{t,e}M'$, if
        \[ M=M''+\sum_{x\in\X(t)}\fmset{m_{e(x)}}
        \qquad
        M'=M''+\out+\sum_{x\in\X(t)}\fmset{m_{e(x)}'}
        \]
    where \(m_{e(x_j)}' = \sum_{x_i \in \X(t)} \left( m_{e(x_i)}'' \ast G_t(x_i,x_j) \right) + F_{x_j}(t,P)\) and \(m_{e(x_i)}'' = m_{e(x_i)} - F_{x_i}(P,t)\).
    We define \(\rightarrow^\ast\) as the reflexive and transitive closure of \(\rightarrow^{t,e}\).
\end{definition}

\begin{example} \label{example:CNuPN2}

Fig.~\ref{fig:CNuPN_step3} depicts the firing of transition $t$ from Fig.~\ref{fig:CNuPN_init} with mode $e(x_1)= a$, $e(x_2)= b$, $e(x_3)= c$, $e(\nu_1)=d$, and $e(\nu_2)=f$. Fig~\ref{fig:CNuPN_step1} and Fig.~\ref{fig:CNuPN_step2} depict the intermediate steps.

\end{example}

The qonCG of a \GnPN is defined similarly to \nPNs.
\begin{definition}[\GnPN qonCG]
    Given a \GnPN 
    $W=\tup{P,T,F,G}$, 
    the \textit{qonCG of $W$} is the tuple $\C=((C,E),\prec,\norm{\bullet})$ where, $C$ is the set of configurations of $W$ and, for each $M_1,M_2\in C$, 
    \begin{enumerate}
        \item $E(M_1,M_2)$ iff there is some transition firing such that $M_1\rightarrow M_2$, and 
    \item $\norm{M_1}=\sum_{m_i \in \support{M}} M(m_i)|m_i|$
    \end{enumerate}
\end{definition}

\section{Phase Encodings}\label{sec:phaseencoding}

The results of this paper rely on a number of encodings among cEOS, \GnPN, and \nPN. In this section we define the type of encoding that we exploit. Specifically, we encode a single step of a source CG into a sequence of steps in a target CG. The first and last configurations in the sequence amount to the actual encoding of the configurations in the source step, while the intermediate configurations amount to extra steps required in the target CG in order to encode the single source step.

\begin{definition}[Configuration encoding]
   Given two CGs $\C_1=(C_1,E_1)$ and $\C_2=(C_2,E_2)$, an \textit{configuration encoding} $\xi$ of $\C_1$ into $\C_2$ is an injective function $\xi:C_1\longrightarrow C_2$. The configurations $K\in C_2$ in the range of $\xi$ are called \textit{encoding configurations}. A \textit{phase} $P$ is either:
   \begin{enumerate}
       \item an infinite run $K_1\rightarrow K_2\rightarrow\dots$ of $\C_2$ such that $K_1$ is an encoding configuration iff $i=1$, or
       \item it is a finite run $K_1\rightarrow K_2\rightarrow\dots\rightarrow K_n$ such that, $n > 1$ and, for each $i>1$, $K_i$ is an encoding configuration iff $i\in\{1,n\}$.
   \end{enumerate} 
   We write $\start(P)=K_1$ and, if $P$ is a finite phase $K_1\rightarrow^* K_n$, then $\finish(P)=K_n$. A phase $P_1$ is consistent with a phase $P_2$ if $P_1$ is finite and $\finish(P_1) = \start(P_2)$.
\end{definition}
Since $\xi$ is an injective function, given a configuration $K\in C_2$ in the range of $\xi$,  is the counter-image $\xi^{-1}(K)$ of $K$ under $\xi$, i.e., $\xi^{-1}(K)=\{H\in C_1\mid \xi(H)=K\}$, contains a single configuration $H\in C_1$. Thus, with a slight abuse of notation, we denote $H$ by $\xi^{-1}(K)$.

In a \textit{phase encoding}, while configurations are encoded into encoding configurations, steps are encoded into corresponding phases.

\begin{definition}[Phase Encoding]
    Given a configuration encoding $\xi$ of $\C_1$ into $\C_2$, 
   $\xi$ is a \textit{phase encoding} if:
   \begin{enumerate} 
        \item if $K_1\rightarrow K_2$ in $\C_1$, then, there is a phase $P$ such that $\start(P)=\xi(K_1)$ and $\finish(P)=\xi(K_2)$, and
        \item if there is a finite phase $P$, then there is a step $\xi^{-1}(\start(P))\rightarrow \xi^{-1}(\finish(P))$ in $\C_1$.
    \end{enumerate} 
    Provided that $\C_1=((C_1,E_1),\leq_1,\norm{\bullet}_1)$ and $\C_2=((C_2,E_2),\leq_2,\norm{\bullet}_2)$ are qonCGs, we say that $\xi$ is:
    \begin{enumerate}
        \item an \textit{embedding} if $K\leq_1 H$ if and only if $\xi(K)\leq_2 \xi(H)$. 
        \item \textit{uniform} if it is an embedding and, if $H\geq_2 \xi(K)$ in $\C_2$ and $H$ is reachable from some encoding configuration, then $H$ is an encoding configuration.
        \item \textit{phase-finite} if all the phases are finite. 
        \item \textit{$f$-bounded} if for each run $\sigma:\xi(K)\rightarrow^*H$ in $\C_2$ where the only encoding configuration in $\sigma$ is $\xi(K)$, we have $\norm{H}_2\leq_2 f(\norm{\xi(K)}_2)$, for some $f:\mathbb{N}\longrightarrow\mathbb{N}$.
        \item $f$-perfect if it is all of the above, for some $f:\mathbb{N}\longrightarrow\mathbb{N}$.
    \end{enumerate}
\end{definition}

Note that the definition of a phase encoding is reminiscent of transitive, stuttering, or branching bisimulations and compatibility properties~\cite{BrowneCG88,finkel_well-structured_2001}. However, while bisimulations are symmetric (they are a simulation from $\C_1$ to $\C_2$ and vice-versa), phase encodings are asymmetric, since they do not simulate single steps of $\C_2$, but only phases.

We now provide a couple of technical lemmas that show that coverability, boundedness, and termination are preserved under appropriate phase encodings. Their proofs are in Appendix.~\ref{app:phaseProofs}. Given a problem instance $I=(\C_1,K_0,\dots,K_n)$, and a phase-encoding $\xi$ of $\C_1$ into some $\C_2$, we denote by $\xi(I)$ the instance $(\C_2,\xi(K_0),\dots,\xi(K_n))$.

\newcommand{\coverLemmaStatement}{
    Given two families $F_1$ and $F_2$ of qonCG, if, for each $\C_1\in F_1$, there is a \textit{uniform phase-encoding} $\xi$ of $\C_1$ into some $\C_2\in F_2$ then, each instance $I$ of $F_1$-coverability is equivalent to the instance $\xi(I)$ of $F_2$-coverability.
}
\begin{lemma}\label{coverLemmaLabel}
    \coverLemmaStatement
\end{lemma}

\newcommand{\terminationLemmaStatement}{
    Given two families $F_1$ and $F_2$ of qonCGs, if, for each $\C_1\in F_1$, there is a \textit{phase-finite phase-encoding} $\xi$ of $\C_1$ into some $\C_2\in F_2$ then, each instance $I$ of $F_1$-termination is equivalent to the instance $\xi(I)$ of $F_2$-termination.
}
\begin{lemma}\label{terminationLemmaLabel}
    \terminationLemmaStatement
\end{lemma}

\begin{definition}
    A qonCG $((C,E),\leq,\norm{\bullet})$ is \textit{finitary} if, for each $n\in\mathbb{N}$, $\norm{\bullet}^{-1}(n)$ is finite.
\end{definition}
Note that the qonCGs of \nPNs, cEOSs, and \GnPNs are all finitary since, given a finite number of tokens, one can build at most a finite number of configurations.

\newcommand{\boundedLemmaStatement}{
   Given two families $F_1$ and $F_2$ of qonCGs, if, for each $\C_1\in F_1$, there is a \textit{$f$-bounded 
    phase-encoding} $\xi$ of $\C_1$ into some \textit{finitary} $\C_2\in F_2$,
    then each instance $I$ of $F_1$-boundedness is equivalent to the instance $\xi(I)$ of $F_2$-boundedness.
}
\begin{lemma}\label{boundedLemmaLabel}
    \boundedLemmaStatement
\end{lemma}

\section{From \texorpdfstring{\nPN}{nuPN} to cEOS}\label{sec:fromNutoCEOS}
We now provide a phase encoding from an arbitrary \nPN to a cEOSs. We then obtain a reduction from problems on \nPNs to the problems on cEOS.

Let $\D=(P,T,F)$ be a \nPN and $\ell=|P|$. We build a cEOS $\os$ that involves only two types: the type $\blacksquare$, used to fire sequences of events, and one more type $N_\D$, that captures the transitions of $\D$ split by variable. 

\begin{definition}
    $N_\D$ is the PN $N_\D=(P_\D,T_\D,F_\D)$ such that $P_\D=P$, $T_\D=\biguplus_{t\in T}\{t_x\mid x\in\var(t)\}$, and, for each $p\in P_\D$ and $t_x\in T_\D$, $F_\D(p,t_x)=F_x(p,t)$ and $F_\D(t_x,p)=F_x(t,p)$.
\end{definition}

Since we deal with only two types, we graphically represent system net places of type $N_\D$ by circles, as usual, while we represent places of type $\blacksquare$ by triangles.

The system net $\hat{N}=(\hat{P},\hat{T},\hat{F})$ of $\os$ contains the places ${sim}$, ${selectTran}$, and, for each $t\in T$, the places, transitions, and flow function depicted in Fig.~\ref{fig:lowerboundnuPNtoEOSblock}. The figure depicts also the synchronization structure $\Theta$ of $\os$. Note that 
\begin{enumerate}
    \item the places $sim$ and $selectTran$ are in common for each instantiation of  Fig.~\ref{fig:lowerboundnuPNtoEOSblock} for $t\in T$, and
    \item the only transitions that can fire concurrently are the $t^\text{fire}_{x}$, for $x\in\var(t)$.
\end{enumerate}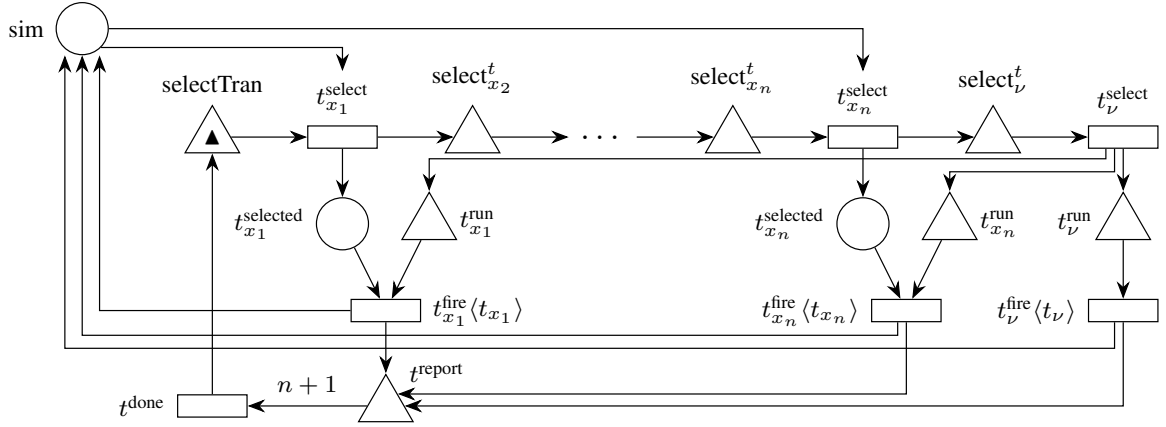
\begin{figure}[t]\centering
\resizebox{1\columnwidth}{!}{
\begin{tikzpicture}

\node[triangle,label={[name=selLab]above:\scriptsize $\text{selectTran}$}](sel)at(0,0){};
\node[transhor,label={[name=t1Lab]above:\scriptsize $t_{x_1}^{\text{select}}$}](t1)at(1.5,0){};
\node[triangle,label={[name=sel1Lab]above:\scriptsize $\text{select}_{x_2}^t$}](sel1)at(3,0){};
\node at(4.5,0)(selDots){$\dots$};
\node[triangle,label={[name=t2Lab]above:\scriptsize $\text{select}_{x_n}^t$}](t2)at(6,0){};
\node[transhor,label={[name=sel2Lab]above:\scriptsize $t_{x_n}^\text{select}$}](sel2)at(7.5,0){};

\node[triangle,label={[name=tnuLab]above:\scriptsize $\text{select}_{\nu}^t$}](tnu)at(9,0){};
\node[transhor,label={[name=selnuLab]above:\scriptsize $t_\nu^\text{select}$}](selnu)at(10.5,0){};

\node[place,label={[name=simLab]left:\scriptsize $\text{sim}$}](sim)at(-1.5,1.25){};

\node[place,label={[name=t1selLab]left:\scriptsize $t^\text{selected}_{x_1}$}](t1sel)at(1.5,-1){};
\node[triangle,label={[name=t1runLab]right:\scriptsize $t^\text{run}_{x_1}$}](t1run)at(2.5,-1){};

\node[transhor,label={[name=t1putLab]right:\scriptsize $t^\text{fire}_{x_1}\tup{t_{x_1}}$}] at (2,-2)(t1put){};

\node[place,label={[name=tnselLab]left:\scriptsize $t^\text{selected}_{x_n}$}](tnsel)at(7.5,-1){};
\node[triangle,label={[name=tnrunLab]right:\scriptsize $t^\text{run}_{x_n}$}](tnrun)at(8.5,-1){};

\node[transhor,label={[name=tnputLab]left:\scriptsize $t^\text{fire}_{x_n}\tup{t_{x_n}}$}] at (8,-2)(tnput){};

\node[triangle,label={[name=tnurunLab]left:\scriptsize $t^\text{run}_{\nu}$}](tnurun)at(10.5,-1){};

\node[transhor,label={[name=tnuputLab]left:\scriptsize $t^\text{fire}_{\nu}\tup{t_\nu}$}] (tnuput)at (10.5,-2){};

\draw[->] (sim.south east) -| (t1Lab.north);
\draw[->] (sim.east) -| (sel2Lab.north);

\draw[->] (t1.south) -- (t1sel.north);
\draw[->] (sel2.south) -- (tnsel.north);

\draw[->] ($(selnu.south)-(.2,0)$) |- ($(t1run)+(0,.75)$) -- (t1run);
\draw[->] ($(selnu.south)-(.1,0)$)|- ($(tnrun)+(0,.6)$) -- (tnrun);
\draw[->] ($(selnu.south)-(0,0)$) -- (tnurun);

\draw[->] (sel) -- (t1);
\draw[->] (t1) -- (sel1);
\draw[->] (sel1) -- (selDots);
\draw[->] (selDots)--(t2);
\draw[->] (t2)--(sel2);
\draw[->] (sel2)--(tnu);
\draw[->] (tnu)--(selnu);

\draw[->] (t1put) -| ($(sim.south)+(.2,0)$);
\draw[->] ($(tnput.south) - (.1,0)$) -- ($(tnput.south)-(.1,.15)$) -| (sim);
\draw[->] ($(tnuput.south) - (.1,0)$)-- ($(tnuput.south)-(.1,.3)$)-| ($(sim.south)-(.2,0)$);

    \node[triangle,label={[name=reportLab]above right:\scriptsize $t^\text{report}$}](report)at (2,-3.1){};

\node[transhor,label={[name=doneLab]left:\scriptsize $t^\text{done}$}](done)at(0,-3.1){};

\draw[->] (t1put) -- (report);
\draw[->] (tnput) |- (report.north east);
\draw[->] (tnuput) |- (report.east);
\draw[->] (report) --node [midway, above]{\scriptsize 
$n+1$} (done);
\draw[->] (done) -- (sel);

\node at (sel){\scriptsize $\blacksquare$};

\draw[->](t1sel) -- (t1put);
\draw[->](t1run) -- (t1put);

\draw[->](tnsel) -- (tnput);
\draw[->](tnrun) -- (tnput);

\draw[->](tnurun) -- (tnuput);

\end{tikzpicture}
}

    \caption{The part of $\os$ dedicated to the simulation of a transition $t$ of $\D$ such that $\var(t)=\{x_1,\dots,x_n,\nu\}$. If $\nu\notin\var(t)$, then, $\text{select}^t_\nu$, $t^\text{select}_\nu$, $t_\nu^\text{run}$, and $t_{\nu}^\text{fire}$ have to be dropped, and $\hat{F}(t^\text{select}_x,t_{x_1}^\text{run})=1$, $\hat{F}(t^\text{select}_x,t_{x_n}^\text{run})=1$, and $\hat{F}(t^\text{report},t^\text{done})=n$ has to be set.
    }
    \label{fig:lowerboundnuPNtoEOSblock}
\end{figure}

Intuitively, given a configuration $M$ of $\D$, we encode $M$ by using, for each tuple $m\in M$ a dedicated $N_\D$ object in the place \textit{sim} with internal marking $m$, plus a single token in $selectTran$.  
\begin{definition}\label{def:lowerBoundEncode}
    Let $CG_{\D}=((C_\D,E_\D),\preceq,\norm{\bullet}_N)$ be the CG of $\D$ and $CG_{\os}=((C_\os,E_\os),\leq_f,\norm{\bullet}_\os)$ the CG of $\os$. The configuration $M=\fmset{m_1,\dots,m_\ell}\in C_{\D}$ is encoded by the configuration $\xi(M)=\sum_{i=1}^\ell \tup{\text{sim}, m_i} + \tup{\text{selectTran}, \varepsilon}$.
\end{definition}
Note that the function $\xi:C_{\D}\longrightarrow C_{\os}$ is injective, i.e., a configuration encoding of $CG_{\D}$ into $CG_{\os}$. Moreover, for each configuration $M$ of $\D$ and marking $m$ over $P$, $\xi(M)+\tup{sim,m}=\xi(M+m)$. By tracing the firing sequences in Fig.\ref{fig:lowerboundnuPNtoEOSblock}, one can show that $\xi$ is a phase encoding that satisfies the properties in Sec.\ref{sec:phaseencoding}.

\newcommand{\nuToCEOSphaseEncodingStatement}{The configuration encoding $\xi$ defined in Def.\ref{def:lowerBoundEncode} is a phase encoding.}
\begin{lemma}\label{lemma:nuToCEOSphaseEncoding}
    \nuToCEOSphaseEncodingStatement
\end{lemma}

The proof is available in App.\ref{app:phaseEncodingProofs}. We now show that $\xi$ is perfect.
\begin{lemma}
    The configuration encoding $\xi$ in Def.\ref{def:lowerBoundEncode} of $CG_N$ into a $CG_{\os}=(\C_\os,\leq_f,\norm{\bullet}_\os)$ is $f$-perfect, for some polynomial function $f$.
\end{lemma}
\begin{proof}
If $H$ is reachable from some encoding configuration $\xi(K)$, then $H$ is reachable via some phase starting from an encoding configuration. Thus, if it has a $\blacksquare$ token on $selectTran$, then $H=\sum_{i=1}^n\tup{sim,m_i}+\tup{selectTran,\varepsilon}$.

\paragraph{Embedding} Clearly, $\xi$ is monotone. Vice-versa, if $\xi(H)\geq_f \xi(K)$, since both $\xi(H)$ and $\xi(K)$ are encoding configurations with a $\blacksquare$ on $selectTran$, they differ only on the tuples on $sim$. Specifically, $\xi(K)=\xi(H)+\sum_{m\in M}\tup{sim,m}=\xi(H+M)$ for some finite set $M$ of tuples over $P$. By injectivity of $\xi$, $K=H+M\succeq H$. Thus, $\xi$ is an embedding.

\paragraph{Uniform}  If $H$ is reachable from an encoding configuration and $H\geq_f\xi(K)$, then $H=\xi(K)+\sum_{m\in M}\tup{sim,M}=\xi(K+\sum_{m\in M}\tup{sim,M})$ for some finite set $M$ of tuples over $P$. Thus, $\xi$ is uniform.

\paragraph{phase-finite} Each phase $\pi$ starts with the firing of a $t^{select}_{x_1}$ event. Now, each transition in Fig.\ref{fig:lowerboundnuPNtoEOSblock} can fire at most once, firing first the $t^{select}_{x}$ events, then the $t^{fire}_{x}$ events, and finally the $t^{done}$ event, which returns a new encoding configuration. For each $t$, these are at most $2|\var(t)|+3$ events. Thus, each phase has at most $2\max_{t\in T}|\var(t)|+3$ steps. Thus, each phase is finite.

\paragraph{$f$-bounded} Since each event in a phase can fire at most once, each phase is long at most $s=2\max_{t\in T}|\var(t)|+3$, and each event creates at most $k=\max_{t\in T, x\in \var(t)}(\sum post(t_{x_i}))+2$, if $H$ is reachable from $K$ during a phase, $\norm{H}_N\preceq \norm{K}_N+sk$.

\end{proof}

Thus, by the lemmas in Sec.\ref{sec:phaseencoding}, $\xi$ preserves the instances of coverability, termination, and boundedness. Since, for each \nPN $N$ and \nPN configuration $M$ the corresponding encoding cEOS $\os$ and encoding configuration $\xi(M)$ can be built in polynomial time, we have reduced the problems over \nPNs to the corresponding problems over cEOSs.
\begin{theorem}\label{thm:npnCEOS}
    There is a polynomial reduction from \nPN-coverability, (respectively, -termination, -boundedness) to cEOS-coverability (-termination, -boundedness).
\end{theorem}

Since \nPN-coverability is hyper-Ackermanian, $\F_{\omega2}$-complete~\cite{LazicS16}, and \nPN-boundedness and -termination are non-primitive recursive~\cite{DBLP:journals/tcs/Rosa-VelardoF11}, we obtain analogous lower-bounds for cEOSs.
\begin{corollary}
    cEOS-coverability is $\F_{\omega2}$-hard. cEOS-termination and -boundedness are non-primitive recursive.
\end{corollary}

\section{From cEOS to \texorpdfstring{\GnPN}{c-nuPN}}\label{sec:fromEOS}

In this section, we show that for any given cEOS \( \os= (\hat{N},\N,d,\Theta)\), there is a polynomial-time constructible \GnPN \(W=(P,T,F,G)\) that simulates \(\os\). We assume without loss of generality that, for each system net transition \(\tau \in \hat{T}\), \(\tau\) is involved in exactly one event.\footnote{If \(\tau\) is involved in multiple events $(\theta_i)_{i=1}^n$, we substitute \(t\) in \(\theta_i\) with a copy \(t_i\) of \(t\), for each $i\in \{1,\cdots,n\}$.}
In what follows, we show only how to encode an arbitrary synchronous event $e=\tup{\tau,\theta}$ according to the EOS semantics,
i.e., by serializing merging, internal firing, and distribution phases. In fact, object and system autonomous events are restricted forms of synchronous events. We assume that $\tau$ involves only two system place types $N,\blacksquare\in\N$.\footnote{The construction for events involving several types is analogous, requiring only the serialization of a copy of the given construction for each type.} The \GnPN $W$ combines several blocks of places:
\begin{enumerate*}
    \item \textbf{\(\hat{p}\)-Block} (see Fig.~\ref{fig:ceosTonuPNBasicBlocks}), for each system net place $\hat{p}$, used to encode the configurations of $\os$;
    \item \textbf{$N$-merged} and \textbf{$N$-updated} blocks (see Fig.~\ref{fig:emerge} and Fig.~\ref{fig:eupdating}), for each type \(N \in \N\setminus\{\blacksquare\}\), consists of some auxiliary places to store the result of the merging phase and conduct the internal firing phase of the event $e$.
\end{enumerate*}
These blocks are utilized in several modules to simulate the firing of an event \(e\) in \(\os\).
\begin{enumerate*}
    \item \textbf{$e$-merging} modules (see Fig.~\ref{fig:emerge}), to capture the dynamics of the merging stage;
    \item \textbf{$e$-updating module} (see Fig.~\ref{fig:eupdating}), to capture the dynamics of the internal firing stage;
    \item \textbf{$e$-distributing module} (see Fig.~\ref{fig:edist}), to capture the dynamics of the internal firing stage and the distribution stage.
\end{enumerate*} 
These modules fire in sequence, in the order above, thanks to the passing of a $\blacksquare$ token handled by a $x_0$ variable.
Overall, these modules can be built in polynomial time, and their consecutive firing allows $W$ to simulate the firings of $e$ in $\os$.  

In the rest of this section, we first define \(\hat{p}\) block and formally state the configuration encoding of the EOS in \cref{subsec:encoding}. Then, in \cref{subsec:phase_encoding_ceos_to_cnpn}, we define \(N\)-merged and \(N\)-updated blocks followed by the definition of all the modules that constitute the phase encoding of a single step of \(\os\) in \(W\), through the example in Fig.~\ref{fig:ceosex}. We conclude the section by proving that the phase encoding defined in \cref{subsec:phase_encoding_ceos_to_cnpn} is \(f\)-perfect in \cref{subsec:f_perfect_ceos_to_cnpn}

\begin{figure}[t]
    \begin{subfigure}{0.45\linewidth}
        \centering
        \scalebox{0.8}{

\begin{tikzpicture}
    \begin{scope}
    \node[place,label={[name=p1obj1Lab]left:\scriptsize$p_1$}]at(0,0)(p1obj1){};
    \node at(p1obj1) {\scriptsize$\bullet\bullet$};
    \node[transhor,label=left:\scriptsize$t$]at(0,-1)(tobj1){};
    \node[place,label=left:\scriptsize$p_2$]at(0,-2)(p2obj1){};
    \node at(p2obj1) {\scriptsize$\bullet$};
    \draw[->] (p1obj1) -- (tobj1);
    \draw[<-] (p2obj1) -- (tobj1);

    \node[draw,dashed,fit={(p1obj1)(tobj1)(p2obj1)(p1obj1Lab)}](obj1){};
    \end{scope}
   \begin{scope}[xshift=1.7cm]
    \node[place,label={ left:\scriptsize$\hat{p}_1$}]at(0,0)(p1sys){};
    \node[xshift=-.1cm] (bul1)at(p1sys) {\scriptsize$\bullet$};
        \node[xshift=.1cm] (bul2)at(p1sys) {\scriptsize$\bullet$};
    \node[transhor,label={left:\scriptsize$\hat{t}$},label={below left:\scriptsize$\fmset{2t}$}]at(0,-1)(tsys){};
    \node[place,label=left:\scriptsize$\hat{p}_2$]at(0,-2)(p2sys){};

    \draw[->] (p1sys) --node[midway,right]{\scriptsize$2$} (tsys);
    \draw[<-] (p2sys) --node[midway,right]{\scriptsize$2$} (tsys);

    \end{scope}

   \begin{scope}[xshift=2*1.7cm]
    \node[place,label={[name=p1obj2lab]left:\scriptsize${p}_1$}]at(0,0)(p1obj2){};
    \node at(p1obj2) {\scriptsize$\bullet$};
    \node[transhor,label=left:\scriptsize$t$]at(0,-1)(tobj2){};
    \node[place,label=left:\scriptsize${p}_2$]at(0,-2)(p2obj2){};
    \node at(p2obj2) {\scriptsize$\bullet\bullet$};
    \draw[->] (p1obj2) -- (tobj2);
    \draw[<-] (p2obj2) -- (tobj2);

    \node[draw,dashed,fit={(p1obj2)(tobj2)(p2obj2)(p1obj2lab)}](obj2){};
    \end{scope}

    \draw[dashed] (obj1)--(bul1.center);
    \draw[dashed] (obj2)--(bul2.center);
\end{tikzpicture}
}
      \caption{}
      \label{fig:ceosex}
\end{subfigure}
\begin{subfigure}{0.45\linewidth}
        \centering
        \scalebox{1}{\begin{tikzpicture}

\path [rectangle, fill=gray!10](2.5,-1.7) to (3.5,-1.7) to (3.5,1.5) to (2.5,1.5);
\node[label={[name=p1BlockLab]above:\scriptsize  $\hat{p}_1$-block}](p1Block)at(3,1.3){};
\node[place,label={[name=idpLab]below:\scriptsize $Id_{\hat{p}_1}$}](idp)at (3,0){};
\node[empty] (txt) at (3,0.1){\scriptsize \tiny $a_1$};
\node[empty] (txt) at (3,-0.1){\scriptsize \tiny $a_2$};
\node[place,label={[name=p1Lab]below:\scriptsize $p_1$}](p1)at (3,1){};
\node[empty] (txt) at (3,1.1){\scriptsize \tiny $a_1 a_1$};
\node[empty] (txt) at (3,0.9){\scriptsize \tiny $a_2$};
\node[place,label={[name=p2Lab]below:\scriptsize $p_2$}](p2)at (3,-1){};
\node[empty] (txt) at (3,-0.9){\scriptsize \tiny $a_1$};
\node[empty] (txt) at (3,-1.1){\scriptsize \tiny $a_2a_2$};

 \path [rectangle, fill=gray!10](5,-1.7) to (6,-1.7) to (6,1.5) to (5,1.5);

\node[label={[name=idpnBlockLab]above: \scriptsize $\hat{p}_2$\text{-block}}](pNBlock)at(5.5,1.3){};
\node[place,label={[name=idpnLab]below:\scriptsize $Id_{\hat{p}_2}$}](idpn)at (5.5,0){};

\node[place,label={[name=p1NLab]below:\scriptsize $p_1$}](p1N)at (5.5,1){};
\node[place,label={[name=p2NLab]below:\scriptsize $p_2$}](p2N)at (5.5,-1){};

\end{tikzpicture}
}     
      \caption{}
      \label{fig:flatten}
  \end{subfigure}
  \caption{A cEOS Example (\subref{fig:ceosex}) and its encoding configuration in two blocks (\subref{fig:flatten}).}
  \label{fig:ceosTonuPNBasicBlocks}
\end{figure}
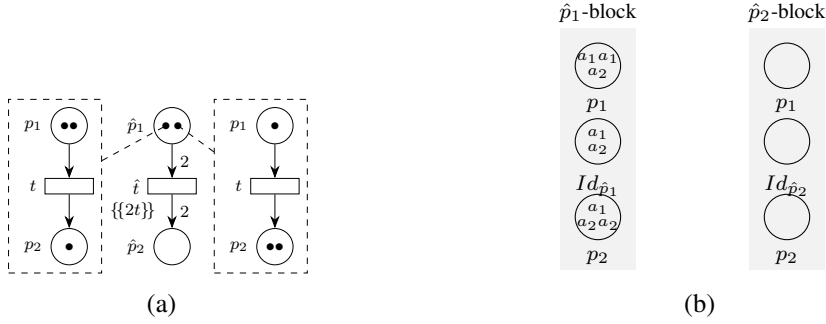

\subsection{cEOS to \texorpdfstring{\nPN}{nuPN} encoding} \label{subsec:encoding}

To formally define the configuration encoding \(\xi\), we need only \textbf{\(\hat{p}\)-Block} next to the place $p^{init}$ in Fig.\ref{fig:emerge}. The other places are used only in the phases and remain empty in each encoding configuration. For each $\hat{p}\in \hat{P}$, the module \textbf{\(\hat{p}\)-Block},  is used to encode in $W$ the objects hosted by $\hat{p}$. Specifically, \textbf{\(\hat{p}\)-Block} contains (a disjoint copy of) the places of $N$ as well as a place $\Id_{\hat{p}}$. 
The latter is used to store identifiers dedicated to each encoded object. 
A configuration $\tup{\hat{p},m}$ of $\os$ is captured by the tuple $m'$ of \textbf{\(\hat{p}\)-Block} that extends $m$ with a token on $\Id_{\hat{p}}$.
Note that, thanks to the place $\Id_{\hat{p}}$, this encoding can even witness the case $m=0$, by placing a single token on $\Id_{\hat{p}}$.

Say that $p$ is a place of $N\in \N$, we denote by $p^{\hat{q}}$ the corresponding place in $\hat{q}$-block for each system net place $\hat{q}$ of type $N$, by $p^{mer}$ the corresponding place in $N$-merged, and by $p^{upd}$ the corresponding place in $N$-updated. The following function $\K$ is used to transform objects in $\os$ as tuples in the modules of $W$.

\begin{definition}
    Given a configuration $c=\tup{\hat{p},m}$ of $\os$ where $d(\hat{p})=N=(P_N,T_N,F_N)$, we define the tuple 
    \begin{align*}
        \K(c)=&\sum_{p\in P_N} m(p)\delta_{p^{\hat{p}}}+\delta_{\Id_{\hat{p}}}&
        &
        \\
        \K_{mer}(c)=&\sum_{p\in P_N} m(p)\delta_{p^{mer}}+\delta_{\Id_N^m}&
        \bar{\K}_{mer}(N,{m})=&\sum_{p\in P_N} m(p)\delta_{p^{\hat{p}}}+\delta_{\Id_{\hat{p}}}\\
        \K_{upd}(c)=&\sum_{p\in P_N} m(p)\delta_{p^{upd}}+\delta_{\Id_N^u}&
        \bar{\K}_{mer}(N,{m})=&\sum_{p\in P_N} m(p)\delta_{p^{\hat{p}}}+\delta_{\Id_{\hat{p}}}
    \end{align*}
\end{definition}

We extend $\K$, $\K_{mer}$, and $\K_{upd}$ to sums of pairs $c_i=\tup{\hat{p}_i,m_i}$ in the trivial way: given a configuration $M=\sum_{i=1}^\ell \tup{\hat{p}_i,m_i}$ of $\os$, letting $c_i=\tup{\hat{p}_i,m_i}$, we define
    \begin{align*}
        \K(M)=&\sum_{i=1}^\ell \fmset{\K(c_i)}\\
        \K_{mer}(M)=&\sum_{i=1}^\ell \fmset{\K_{mer}(c_i)}\\
        \K_{upd}(M)=&\sum_{i=1}^\ell \fmset{\K_{upd}(c_i)}\\
    \end{align*}
We can now define the encoding of an arbitrary configuration.
\begin{definition} \label{def:encoding_ceos_to_cnpn}
    The encoding configuration $\xi(M)$ of a configuration $M$ of $\os$ is the configuration $\xi(M)=\K(M)+\fmset{\delta_{p^{init}}}$ of $W$.
\end{definition}

Note that the function $\xi:C_{\os}\longrightarrow C_{W}$ is injective, i.e., a configuration encoding of $CG_{\os}$ into $CG_{W}$.

\subsection{Phase Encoding} \label{subsec:phase_encoding_ceos_to_cnpn}

We now come back to the construction, where we show the modules in \GnPN \(W\) that constitute a phase encoding of a single step in \(\os\) through the running example in \cref{fig:ceosex}.

\smallskip
\noindent\textbf{\textbf{$N$-merged} and \textbf{$N$-updated} blocks.}
These blocks are used to store the objects after the merging and internal firing stages and are analogous to the blocks above. 
The \textbf{$N$-merged} module and the \textbf{$N$-updated} modules contain a disjoint copy of the places of $N$
as well as the place $\Id_{N}^m$ and $\Id_{N}^u$, respectively. 
The encoding of a merged/updated tuple is similar to that in \textbf{$\hat{p}$-Block}.

\smallskip
\begin{figure}[t]
        \centering
\scalebox{1}{\begin{tikzpicture}

\node[place,label={[name=objenbLab]below:\scriptsize $p^{init}$}](pmerge)at (5.5,-1.25){};
\node[place,label={[name=objenbLab]above:\scriptsize $p^{merged}_{e}$}](pmerged)at (5.5,1.25){};
\path [rectangle, fill=gray!10](2.5,-1.7) to (3.5,-1.7) to (3.5,1.5) to (2.5,1.5);
\node[label={[name=p1BlockLab]above:\scriptsize  $\hat{p}_1$-block}](p1Block)at(3,1.3){};
\node[place,label={[name=idpLab]below:\scriptsize $Id_{\hat{p}}$}](idp)at (3,0){};
\node[empty] (txt) at (3,0.1){\scriptsize \tiny $a_1$};
\node[empty] (txt) at (3,-0.1){\scriptsize \tiny $a_2$};
\node[place,label={[name=p1Lab]below:\scriptsize $p_1$}](p1)at (3,1){};
\node[empty] (txt) at (3,1.1){\scriptsize \tiny $a_1 a_1$};
\node[empty] (txt) at (3,0.9){\scriptsize \tiny $a_2$};
\node[place,label={[name=p2Lab]below:\scriptsize $p_2$}](p2)at (3,-1){};
\node[empty] (txt) at (3,-0.9){\scriptsize \tiny $a_1$};
\node[empty] (txt) at (3,-1.1){\scriptsize \tiny $a_2a_2$};

\node[transvert,minimum height=12mm,minimum width=5mm] (tmerge)at (5.5,0){};
\node [rotate=-90]at(5.5,0){\tiny$\tau^{merge}_e$};

 \path [rectangle, fill=gray!10](7,-1.7) to (8,-1.7) to (8,1.5) to (7,1.5);

\node[label={[name=idpnBlockLab]above:\tiny N\text{-Merged}}](pNBlock)at(7.5,1.3){};
\node[place,label={[name=idpnLab]below:\scriptsize $\Id_N^m$}](idpn)at (7.5,0){};

\node[place,label={[name=p1NLab]below:\scriptsize $p_1$}](p1N)at (7.5,1){};
\node[place,label={[name=p2NLab]below:\scriptsize $p_2$}](p2N)at (7.5,-1){};

\draw[->](idp)--node[above,near start, yshift = -0.2]{\scriptsize $x_1^{\hat{p}_1},x_2^{\hat{p}_1}$}(tmerge);
\draw[->,myDouble] (p1) -- node[above,midway,sloped] {\scriptsize $x_1^{\hat{p}_1},x_2^{\hat{p}_1}$}(tmerge);
\draw[->,triple] (p2) -- node[below,midway,sloped] {\scriptsize $x_1^{\hat{p}_1},x_2^{\hat{p}_1}$}(tmerge);

\draw[->](tmerge) -- node[left,midway] {\scriptsize $x_{cs}$}(pmerged);

\draw[->,myDouble] (tmerge) -- node[above,midway,sloped] {\scriptsize $x_1^{\hat{p}_1},x_1^{\hat{p}_1}$}(p1N);
\draw[->] (tmerge) to (idpn);
\node[empty] (txt) at (6.7,0.3){\scriptsize $x_1^{\hat{p}_1}$};
\draw[->,triple] (tmerge) -- node[below,midway,sloped] {\scriptsize $x_1^{\hat{p}_1},x_1^{\hat{p}_1}$}(p2N);

\draw[<-](tmerge.south) -- node[midway,left]{\scriptsize $x_{cs}$}(pmerge.north);

\end{tikzpicture}
}     
      \caption{Module $e$-merging.}
      \label{fig:emerge}
\end{figure}
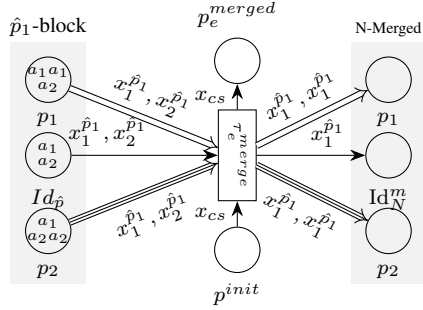
\noindent\textbf{{Module \textbf{\(e\)-merging}}.}
This module captures the dynamics of the merging phase. It contains the places $p^{init}$ and $p^{merged}_e$ as well as the transition $\tau^{merge}_e$. The two places behave as enable and acknowledgment for the transition. Thus, initially, $p^{init}$ is populated by a dedicated \textit{control sequence} token (moved around by the variable $x_{cs}$).
In \(\os\), for each place $\hat{p}$ from which $\tau$ consumes, there is a channel from each non-$\Id$ place of \textbf{$\hat{p}$-Block} to the corresponding place in \textbf{$N$-merged}. 
The pre-channels from $\hat{p}$ are labeled by $x_1^{\hat{p}},\dots,x_n^{\hat{p}}$, where $n$ is the the number of objects consumed by $\tau$ from $\hat{p}$ in \(\os\). The corresponding post-channel is labeled by the constant sequence $x_1^{\hat{p}},\dots,x_1^{\hat{p}}$ of length $n$. 
Moreover, to take care of the identifiers of the encoded objects consumed by \(\tau\), $\tau^{merge}_e$ consumes from $\Id_{\hat{p}}$ one token for each $x_1^{\hat{p}},\dots,x_n^{\hat{p}}$ and produces only one token $x_1^p$ in $\Id_N^m$, where $p$ is the minimum of the places from which $\tau$ consumes according to some fixed order among the system net places.

\smallskip
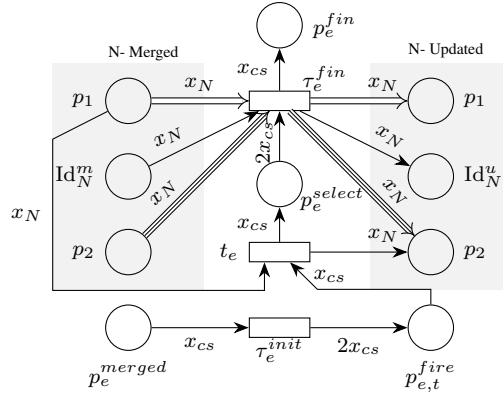
\begin{figure}[t]
        \centering
\scalebox{1} {
\begin{tikzpicture}
\path [rectangle, fill=gray!10](8,-1.5) to (10,-1.5) to (10,1.5) to (8,1.5);
\node[label={[name=idpnBlockLab]below:\scriptsize \tiny N\text{- Merged}}](pNBlock)at(9.2,2){};
\node[place,label={[name=idpnLab]left:\scriptsize $\Id_N^m$}](idpn)at (9,0){};
\node[place,label={[name=p1NLab]left:\scriptsize $p_1$}](p1N)at (9,1){};
\node[place,label={[name=p2NLab]left:\scriptsize $p_2$}](p2N)at (9,-1){};

\node[place,label={[name=Lab]right:\scriptsize $p_e^{fin}$}](pfin)at (11,2){};

\node[transhor] (tfin)at (11,1){};
\node[empty] (txt) at (11.6,1.3){\scriptsize$\tau_{e}^{fin}$};

\node[place](psel)at (11,-0.1){};
\node[empty] (txt) at (11.7,-0.3){\scriptsize$p_{e}^{select}$};
\node[transhor,label={[name=tNLab]left:\scriptsize  $t_e$}] (t)at (11,-1){};
\node[place,label={[name=Lab]below:\scriptsize $p_e^{merged}$}](pinit)at (9,-2){};
\draw[->](t)--node[left,midway] {\scriptsize $ x_{cs}$}(psel);
\draw[->](psel)--(tfin);
\node[empty,rotate=90](txt) at (10.8,0.45){\scriptsize $2x_{cs}$};
\draw[->](tfin)--node[left,midway] {\scriptsize $x_{cs}$}(pfin);
\draw[->,myDouble] (p1N) -- node[above,midway] {\scriptsize $x_{N}$}(tfin);
\draw[->,triple] (p2N) -- node[above,near start,sloped] {\scriptsize $x_{N}$}(tfin);

\path [rectangle, fill=gray!10](12.2,-1.5) to (14,-1.5) to (14,1.5) to (12.2,1.5);
\node[label={[name=BlockLab]above:\scriptsize \tiny N\text{- Updated}}](NUBlock)at(13.2,1.3){};
\node[place,label={[name=idpnLab]right:\scriptsize $\Id_N^u$}](idpu)at (13,0){};
\node[place,label={[name=Lab]right:\scriptsize $p_1$}](p1u)at (13,1){};
\node[place,label={[name=Lab]right:\scriptsize $p_2$}](p2u)at (13,-1){};
\draw[->,myDouble] (tfin) -- node[above,near end] {\scriptsize $x_{N}$}(p1u);
\draw[->,triple] (tfin) -- node[above,near end,sloped] {\scriptsize $x_{N}$}(p2u);

\draw[->](idpn)--node[above,near start,sloped] {\scriptsize $x_{N}$}(tfin);
\draw[->](tfin)--node[above,near end,sloped] {\scriptsize $x_{N}$}(idpu);
\draw[->](p1N) --($(p1N.south)+(-1,-0.2)$) -- node[left,pos=0.5] {\scriptsize $x_{N}$}($(p1N.south)+(-1,-2.2)$) -| ($(t.south)-(0.2,0)$);
\draw[->](t)--node[above,near end] {\scriptsize $x_{N}$}(p2u);
\node[transhor] (tinit)at (11,-2){};
\node[empty] (txt) at (11,-2.3){\scriptsize$\tau_{e}^{init}$};

\node[place,label={[name=Lab]below:\scriptsize $p_{e,t}^{fire}$}](pfire)at (13,-2){};
\draw[->](pinit)--node[below]{\scriptsize $x_{cs}$}(tinit);
\draw[->](tinit)--node[below]{\scriptsize $2 x_{cs}$}(pfire);
\draw[->](pfire)--($(pfire.north)+(0,0.2)$)--($(pfire.north)+(-1.5,0.2)$)--node[right]{\scriptsize$x_{cs}$}(t);
\end{tikzpicture}
}   
      \caption{Module $e$-updating.}
      \label{fig:eupdating}
\end{figure}
\noindent\textbf{Module \textbf{\(e\)-updating}.}
This module captures the dynamics of the internal firing phase applied to the \textbf{$N$-merged} modules after the merging phase, storing its result in the module \textbf{$N$-updated}. It contains the places $p_e^{select}$ as well as a set $\{p^{x}_{e,t} \mid x\in\{fire,fired\}, t \in \supp(\theta)\}$ of new places. The modules contains the transitions $\tau^{init}_e$ and $\tau^{fin}_e$ as well as a copy $t_e$ of the transitions $t\in\supp(\theta)$.
Places $p^{init}_e$, $p^{select}_e$, and $p^{fin}_e$ are used to fire $\tau^{init}_e$ and $\tau^{fin}_e$ in sequence.
The places $p^\text{fire}_{e,t}$ and $p^\text{fired}_{e,t}$ behave as enable to fire and acknowledgment of firing of $t$, respectively.
For each $t\in\supp(\theta)$, $\tau^{init}_e$ puts $\theta(t)$ tokens in $p^{fire}_{e,t}$. Thus, the transition $t_e$ in the module gets enabled $\theta(t)$ times. When firing, $t_e$ behaves as $t$, with the provision that its preconditions stem from places in \textbf{$N$-merged}, while its post-conditions are directed towards places in \textbf{$N$-updated}. The movement of tokens from \textbf{$N$-merged} to \textbf{$N$-updated} ensures that either the transitions in $\theta$ are all enabled at the start of this phase, or a deadlock is reached, i.e., the firing of a $t'\in\theta$ is not responsible for the enabling of another transition $t''\in\theta$.
Finally, all tokens remaining in \textbf{$N$-merged} are moved to the respective places in \textbf{$N$-updated} using the channels of $\tau^{fin}_e$, as well as the movement of the token from $\Id_{N}^m$ to $\Id_{N}^u$, for each $\hat{p}$ involved in the preconditions of $\tau$.

\begin{figure}[t]
    \begin{subfigure}{0.2\linewidth}
        \centering
        \scalebox{0.7}{
\begin{tikzpicture}

\node[place,label=below:\scriptsize$p^{fin}_e$]at(0,-2)(pfin){};
\node[transhor,label=left:\scriptsize$t_e^{id}$]at(0,-1)(t){};
\node[place,label=above:\scriptsize$p^{move(1)}_e$]at(-1,0)(pmove){};
\node[place,label=above:\scriptsize$p^{new}_e$]at(1,0)(pnew){};

\draw[->] (pfin)--node[midway,left]{\scriptsize$x_{cs}$} (t);
\draw[->] (t)-- node[midway,right]{\scriptsize$\nu_1\nu_2$}(pnew);
\draw[->] (t)--node[midway,left]{\scriptsize$x_{cs}$} (pmove);
    
\end{tikzpicture}
} 
      \caption{Module $e$-id-creation.}
      \label{fig:eidcreation}
    \end{subfigure}
    \begin{subfigure}{0.48\linewidth}
        \centering
\scalebox{0.7}{\begin{tikzpicture}

\path [rectangle, fill=gray!10](3.5,-1.7) to (4.5,-1.7) to (4.5,1.5) to (3.5,1.5);
\node[label={[name=BlockLab]above: \tiny N-Updated}](NBlock)at(4,1.3){};
\node[place,label={[name=idpLab]below:\scriptsize $Id_{N}^u$}](idN)at (4,1){};

\node[place,label={[name=p1Lab]below:\scriptsize $p_1$}](p1N)at (4,0){};
\node[place,label={[name=p2Lab]below:\scriptsize $p_2$}](p2N)at (4,-1){};

\node[transvert,label={[name=tLab]below:\scriptsize  $t_1$}](t1)at(5.5,0){};
\node[transvert,label={[name=tLab]below:\scriptsize  $t_2$}](t2)at(6,-1){};
 \path [rectangle, fill=gray!10](6.5,-1.7) to (7.5,-1.7) to (7.5,1.5) to (6.5,1.5);

\node[label={[name=pBlockLab]above:\tiny$\hat{p}_2$-block}](pBlock)at(7,1.3){};
\node[place,label={[name=idpnLab]left:\scriptsize $\Id_{\hat{p}_2}$}](idp)at (7,1){};

\node[place,label={[name=p1NLab]right:\scriptsize $p_1$}](p1p)at (7,0){};
\node[place,label={[name=p2NLab]right:\scriptsize $p_2$}](p2p)at (7,-1){};

\draw[->](p2N)--node[above, near start]{\scriptsize $x_N$}(t2);
\draw[->](t2)--node[above,near end]{\scriptsize $x_{Id}$}(p2p);
\draw[->](p1N)--node[above,near end]{\scriptsize $x_N$}(t1);
\draw[->](t1)--node[above,near end]{\scriptsize $x_{Id}$}(p1p);


\node[place](prename)at (8.5,1){};
\node[empty](txt) at (8.3,1.4){\scriptsize $p^{rename}_{e}$};
\node[place,label={[name=Lab]below:\scriptsize $p^{new}_{e}$}](pnew)at (8.5,-1){};

\node[transhor,label={[name=tLab]below:\scriptsize  $t_3$}](t3)at(9.5,0){};

\node[place](pmoving)at (10.5,1){};
\node[empty](txt) at (11,1.5){\scriptsize $p^{moving(1)}_{e}$};
\node[place,label={[name=Lab]below:\scriptsize $p^{move(1)}_{e}$}](pmove)at (10.5,-1){};
\node[transhor,label={[name=tLab]right:\scriptsize  $t^{moved(1)}_e$}](tmoved)at(9.5,2.5){};

\node[place,label={[name=Lab]left:\scriptsize $p^{transfer}_{e}$}](ptrans)at (7,2.5){};
\draw[->](pnew)--node[above,midway,sloped]{\scriptsize $x_{Id}$}(t3);
\draw[->](t3)--node[above,midway,sloped]{\scriptsize $x_{Id}$}(prename);
\draw[->](prename)--node[below,midway,sloped]{\scriptsize $x_{Id}$}(tmoved);

\draw[->](pmove)--node[above,midway,sloped]{\scriptsize $x_{cs}$}(t3);
\draw[->](t3)--node[above,midway,sloped]{\scriptsize $x_{cs}$}(pmoving);
\draw[->](pmoving)--node[below,midway,sloped]{\scriptsize $x_{cs}$}(tmoved);

\draw[<->]($(prename.south)+(-0.3,0.2)$)--node[left]{\scriptsize $x_{Id}$}($(prename.south)-(0.3,0.1)$)-|(t1.north);

\draw[<->](prename.south)--node[left]{\scriptsize $x_{Id}$}($(prename.south)-(0,1.1)$)-|(t2.north);

\draw[->]($(tmoved.south)-(0.3,0)$)|-node[below,near end]{\scriptsize $x_{Id}$}($(tmoved.south)-(1.8,0.1)$)|-($(idp)+(0.5,0)$)--(idp.east);
\draw[->](tmoved)--node[above,midway,sloped]{\scriptsize $x_{cs}$}(ptrans);
\end{tikzpicture}
}
        \caption{Module $e$-move(1).}
        \label{fig:emove}
    \end{subfigure}  
    \begin{subfigure}{0.3\linewidth}
        \centering
\scalebox{0.7}{\begin{tikzpicture}

\node[place,label={[name=objenbLab]below:\scriptsize $p^{transfer}_e$}](pmerge)at (5,-1.25){};
\node[place,label={[name=newLab]below:\scriptsize $p^{new}_e$}](pnew)at (6,-1.25){};
\node[place,label={[name=objenbLab]above:\scriptsize $p^{init}$}](pmerged)at (5.5,1.25){};
\path [rectangle, fill=gray!10](3.5,-1.7) to (4.5,-1.7) to (4.5,1.5) to (3.5,1.5);
\node[label={[name=BlockLab]above:\tiny N-Updated}](p1Block)at(4,1.3){};
\node[place,label={[name=idpLab]below:\scriptsize $Id_{N}^u$}](idp)at (4,0){};
\node[place,label={[name=p1Lab]below:\scriptsize $p_1$}](p1)at (4,1){};
\node[place,label={[name=p2Lab]below:\scriptsize $p_2$}](p2)at (4,-1){};
\node[transvert,minimum height=12mm,minimum width=5mm] (tmerge)at (5.5,0){};
\node [rotate=-90]at(5.5,0){\tiny$\tau^{transfer}_e$};
\path [rectangle, fill=gray!10](6.5,-1.7) to (7.5,-1.7) to (7.5,1.5) to (6.5,1.5);
\node[label={[name=idpnBlockLab]above:\scriptsize \tiny $\hat{p}_2$\text{-block}}](pNBlock)at(7,1.3){};
\node[place,label={[name=idpnLab]below:\scriptsize $\Id_{\hat{p}_2}$}](idpn)at (7,0){};

\node[place,label={[name=p1NLab]below:\scriptsize $p_1$}](p1N)at (7,1){};
\node[place,label={[name=p2NLab]below:\scriptsize $p_2$}](p2N)at (7,-1){};

\draw[->](idp)--node[above,near start]{\scriptsize $x_{N}$}(tmerge);
\draw[->,myDouble] (p1) -- node[above,midway,sloped] {\scriptsize $x_{N}$}(tmerge);
\draw[->,triple] (p2) -- node[above,midway,sloped] {\scriptsize $x_{N}$}(tmerge);

\draw[->](tmerge) -- node[right,midway] {\scriptsize $x_{cs}$}(pmerged);

\draw[->,myDouble] (tmerge) -- node[above,midway,sloped] {\scriptsize $x_{Id}$}(p1N);
\draw[->] (tmerge) to node[midway, above]{\scriptsize$x_{Id}$} (idpn);
\draw[->,triple] (tmerge) -- node[above,midway,sloped] {\scriptsize $x_{Id}$}(p2N);

\draw[<-](tmerge.south west) -- node[midway,left]{\scriptsize $x_{cs}$}(pmerge.north);

\draw[->](pnew) -- node[midway,right]{\scriptsize$x_{Id}$} (tmerge.south east);

\end{tikzpicture}
}
        \caption{Module $e$-transfer.}
        \label{fig:etransfer}
    \end{subfigure}
    \caption{Module \(e\)-distributing.}
    \label{fig:edist}
\end{figure}

\noindent\textbf{Module \textbf{\(e\)-distributing}.}
This module captures the dynamics of the final distribution of updated tokens.
Let $n$ be the number of objects created by $\tau$ and let $\hat{p}_1,\dots,\hat{p}_n$ be an enumeration (possibly) with repetition of the places in the post-conditions of $\tau$. The module concatenates the sub-module \textbf{$e$-id-creation}, a sequence of sub-modules \textbf{$e$-move$(i)$} for $i\in\{1,\dots,n-1\}$, and the sub-module \textbf{$e$-transfer}. The sub-module \textbf{$e$-id-creation} generates $n$ new identifiers\footnote{This is the only place where we need name creation via $\nu$ variables.} for the encoding of the objects to be created. Then, each \textbf{$e$-move$(i)$} selects one such identifier, say $a$, and moves while renaming to $a$, one by one, some token not in $\Id_{N}^u$ from \textbf{$N$-updated} to the corresponding place in \textbf{$\hat{p}_i$-block}. By taking advantage of concurrency and of a couple of enable/acknowledge places as in the previous modules, each \textbf{$e$-move$(i)$} passes the turn to the next sub-module after a non-deterministic number of movements. 
While passing the turn, \textbf{$e$-move$(i)$} finally moves the identifier $a$ to $\Id_{\hat{p}_i}$, completing the encoding of the next object created by $\tau$ at place $\hat{p}_i$. Finally, the \textbf{$e$-transfer} sub-module creates the last encoding by moving, via channels, all remaining tokens in \textbf{$N$-updated} to \textbf{$\hat{p}_n$-block} while renaming to the last new identifier.

\begin{remark}
    The construction above (or slight modifications) cannot simulate non-conservative EOSs. In general EOS, if the transition $\tau$ in the event $e$ destroys a type $N$, then $e$ can fire only when destroying an empty object, performing a sort of zero-check. Otherwise, the transition cannot fire. Instead, the corresponding module $e$-merging, can fire even when consuming non-empty objects. Thus, this construction works only for cEOSs.
\end{remark}


\newcommand{\phaseencodingeostocnpnstatement}{The configuration encoding $\xi$ in Def.\ref{def:encoding_ceos_to_cnpn} is a phase-encoding, i.e, given an EOS $\os$, two markings $M,M'$, an event $e =\tup{\tau,\theta}$, and mode $(\lambda,\rho)$, $M \rightarrow^{(e,\lambda,\rho)}M'$ iff $\xi(M) \rightarrow^\ast \xi(M')$ in the \GnPN $N$ using the above construction.}
\begin{lemma} \label{thm:phase_encoding_eos_to_cnpn}
\phaseencodingeostocnpnstatement
\end{lemma}

We provide the proof in the Appendix \cref{app:phaseEncodingProofs}. Thanks to Lemma~\ref{thm:phase_encoding_eos_to_cnpn}, \(\xi\) is a phase encoding. Now, we will show that \(\xi\) is \(f\)-perfect.

\subsection{Perfect cEOS to \texorpdfstring{\nPN}{nuPN} encoding} \label{subsec:f_perfect_ceos_to_cnpn}
\begin{lemma}
    The configuration encoding $\xi$ in Def.\ref{def:encoding_ceos_to_cnpn} is $f$-perfect, for some polynomial function $f$.
\end{lemma}
\begin{proof}
If $H$ is reachable from some encoding configuration, then it is reachable via some phase. Thus, if it has a token on $\fmset{\delta_{p^{init}}}$, then $H=\K(M) + \fmset{\delta_{p^{init}}}$.

\paragraph{Embedding} Clearly, if \(H \geq K\) then \(\xi(H) \geq \xi(K)\), since $\xi$ is an monotone. For the other direction, if $\xi(H)\geq \xi(K)$, since both $\xi(H)$ and $\xi(K)$ are encoding configurations with a token \(\fmset{\delta_{p^{init}}}\), they differ only on the tuples on $\K(M)$. Specifically, $\xi(K)=\xi(H)+M'=\xi(H+M')$ for some finite set $M$ of tuples. By injectivity of $\xi$, $K=H+M'\geq H$. Thus, $\xi$ is an embedding.

\paragraph{Uniform}  If $H$ is reachable from an encoding configuration and $H\geq\xi(K)$, then $H=\xi(K)+M=\xi(K+M)$ for some finite set $M$ of tuples. Thus, $\xi$ is uniform.

\paragraph{phase-finite} Say, \(H \rightarrow^t K\) in \(W\). 
We show that \(\xi(H) \rightarrow^\ast \xi(K)\) in polynomial many steps. 
Each phase $\pi$ fires a sequence of five modules.
In module \(e\)-merging, one transition \(\tau^{merge}_{e}\) is fired (see Fig~\ref{fig:emerge}). 
In module \(e\)-updating, three transitions are fired (see Fig~\ref{fig:eupdating}). 
Then, in module \(e\)-distributing, the following three observations can be made: \begin{enumerate*}
    \item Module \(e\)-id-creation, fires only one transition \(t^{id}_e\).
    \item Module \(e\)-move has \(m\) components, where \(m = \sum_x{p \in \hat{P}} \hat{F}(t,p)\) and in each module \(t_3\) and \(t^{moved}_e\) are fired exactly once. The transitions \(t_1\) and \(t_2\) are collectively fired among all the components at most \(k\) times, where \(k\) is the total number of tokens present in \(H\) of \(\os\), i.e., \(|H|\). Hence, there can be at most \(m \times 2 + |H|\) firings in this module.
    \item Module \(e\)-transfer fires \(\tau^{transfer}_e\) only once.
\end{enumerate*}
Clearly, \(\xi\) is phase-finite; in fact, it requires a run of length polynomial to the size of the configuration \(H\) to reach \(\xi(K)\).

\paragraph{$f$-bounded} Since there are polynomial many firings in each phase. And each firing can only add tokens at most linear to the size of the \(\os\), there can only be polynomial many tokens added during a phase. Clearly, \(\xi\) is \(f\)-bounded.

\end{proof}

Thus, by the lemmas in Sec.\ref{sec:phaseencoding}, $\xi$ preserves the instances of coverability, termination, and boundedness. Since, for each cEOS $\os$ and cEOS configuration $M$ the corresponding encoding \GnPN $N$ and encoding configuration $\xi(M)$ can be built in polynomial time. Overall, we have reduced the problems over cEOS to the corresponding problems over \GnPN.
\begin{theorem}\label{thm:ceosGnPN}
    There is a polynomial reduction from cEOS-coverability (respectively, -termination, -boundedness) to \GnPN-coverability (-termination, -boundedness).
\end{theorem}

Since \GnPNs are a fragment of Unordered Data Nets, we may obtain upper-bounds on the complexity of problems on cEOSs. However, in the next section we directly obtain completeness for coverability and more informative upper bounds.

\section{From \texorpdfstring{\GnPN}{c-nuPN} to \texorpdfstring{\nPN}{nuPN}} \label{extendedNuPntoNuPn}

We now provide a configuration encoding of arbitrary \GnPN $W=(P,T,F,G)$ into a standard \nPN $N=(P',T',F')$. Unlike the previous constructions, the encoding we provide is not uniform nor phase-finite, in general. This is because $N$ captures the channels of $W$ by cheating, i.e., it moves the tokens one per time, possibly wrongly leaving some of them unmoved, resulting in a form of \textit{broken} configuration. Our technique unfolds in two steps. First, we abruptly remove from $N$ all runs that reach such broken configurations. This returns a qonCG where the encoding is actually a phase-encoding. Second, we show that putting the broken configurations back in the qonCG does not affect the decision problem under study.
In this section we call the qonCGs of $W$ and $N$ respectively $CG_W$ and $CG_N$.

\subsection{\texorpdfstring{\GnPN}{CnuPN} to \texorpdfstring{\nPN}{nuPN} encoding}

\begin{figure}[t]
    \centering
    \begin{subfigure}[b]{.49\textwidth}\centering
    \scalebox{1}{
    \begin{tikzpicture}

\node[place,label={[name=p1Lab]left:\scriptsize $p_1$}](p1)at (0,0){};
\node[place,label={[name=p2Lab]above:\scriptsize $p_2$}](p2)at (0,1.25){};

\node[place,label={[name=p3Lab]above:\scriptsize $p_3$}](p3)at (1,1.25){};

\node[transvert] (t)at (1.5,0){};
\node at (t){\scriptsize $t$};

\node[place,label={[name=p4Lab]above:\scriptsize $p_5$}](p4)at (3,1.25){};
\node[place,label={[name=p5Lab]right:\scriptsize $p_6$}](p5)at (3,0){};

\draw [->,myDouble] (p1) -- node[below,midway,sloped] {\scriptsize $x_1$}(t);

\node [place,label={above:\scriptsize $p_4$}](p6)at(2,1.25){};
\draw[<->] (p6) -- node[midway,below,sloped]{\scriptsize $x_2$}(t);
\node at(p3){\scriptsize $a$};

\draw [->] (p2) -- node[below,midway,sloped] {\scriptsize $x_3$}(t);
\draw [<->] (p3) -- node[below,midway,sloped] {\scriptsize $x_1$}(t);

\draw [->,myDouble] (t) -- node[below,midway,sloped] {\scriptsize $x_2$}(p5);
\draw [->] (t) --node[below,midway,sloped] {\scriptsize $x_3$} (p4);

\node at($(p1)+(0,.1)$){\scriptsize a a};
\node at($(p1)-(0,.1)$){\scriptsize b c};
\node at($(p2)+(0,.1)$){\scriptsize a b};
\node at($(p2)-(0,.1)$){\scriptsize  c c};
\node at($(p6)+(0,.1)$){\scriptsize c c};
\node at($(p6)-(0,.1)$){\scriptsize d};
    \end{tikzpicture}
    }
\caption{}
\label{fig:selectiveTransferPre}
    \end{subfigure}
    \hfill
       \begin{subfigure}[b]{.49\textwidth}\centering
       \scalebox{1}{
    \begin{tikzpicture}

\node[place,label={[name=p1Lab]left:\scriptsize $p_1$}](p1)at (0,0){};
\node[place,label={[name=p2Lab]above:\scriptsize $p_2$}](p2)at (0,1.25){};

\node[place,label={[name=p3Lab]above:\scriptsize $p_3$}](p3)at (1,1.25){};

\node [place,label={above:\scriptsize $p_4$}](p6)at(2,1.25){};
\node at(p3){\scriptsize $a$};
\node[transvert] (t)at (1.5,0){};
\node at (t){\scriptsize $t$};

\node[place,label={[name=p4Lab]above:\scriptsize $p_5$}](p4)at (3,1.25){};
\node[place,label={[name=p5Lab]right:\scriptsize $p_6$}](p5)at (3,0){};

\draw [->,myDouble] (p1) -- node[below,midway,sloped] {\scriptsize $x_1$}(t);
\draw [->] (p2) -- node[below,midway,sloped] {\scriptsize $x_3$}(t);
\draw [<->] (p3) -- node[below,midway,sloped] {\scriptsize $x_1$}(t);
\draw [<->] (p6) -- node[below,midway,sloped] {\scriptsize $x_2$}(t);

\draw [->,myDouble] (t) -- node[below,midway,sloped] {\scriptsize $x_2$}(p5);
\draw [->] (t) --node[below,midway,sloped] {\scriptsize $x_3$} (p4);

\node at (p5){\scriptsize d d};
\node at (p1){\scriptsize b c};
\node at($(p2)+(0,.1)$){\scriptsize a b};
\node at($(p2)-(0,.1)$){\scriptsize  c};
\node at(p4){\scriptsize  c};
\node at($(p6)+(0,.1)$){\scriptsize c c };
\node at($(p6)-(0,.1)$){\scriptsize d};
    \end{tikzpicture}
    }
\caption{}
\label{fig:selectiveTransferPost}
    \end{subfigure}
    \caption{
    Firing of a simiplified \GnPN transition.
    }
    \label{fig:selectiveTransfer}
\end{figure}
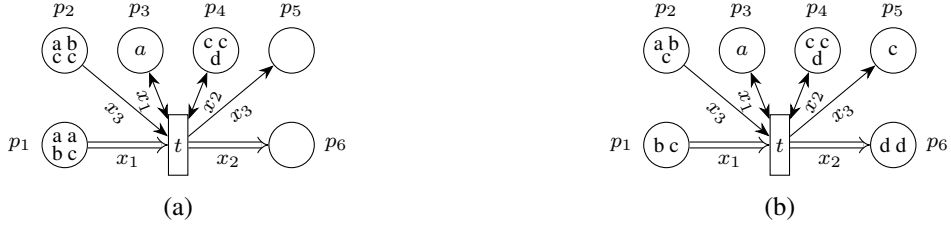
Without loss of generality, we assume that the special transitions of $W$ have the simplified form depicted in Fig.\ref{fig:selectiveTransferPre}. In fact, any arbitrary \GnPN transition can be captured by a sequence of such simplified transitions (cfr. Section 3.2 in \cite{OurRP25}).

The \nPN $N$ contains the same places as $W$, with the addition of some extra place necessary to encode stages, \GnPN modes, and a mechanism to designate \textit{broken} configurations. Specifically, let $T=T_{std}\cup T_{spl}$, where $T_{std}$ (respectively, $T_{spl}$) is the set of standard \nPN transitions (special \GnPN transitions, i.e., with at least one channel). Assuming that, for each transition $t$, $\var(t)=\{x_1,\dots,x_{m_t}\}$ for some $m_t$, let $n=\max_{t\in T}\norm{\var(t)}$, we have

\[ P' = P \cup \{\pactive\} \cup \{p_{\nu}\} \cup \{\pinit{t}, \pfire{t} \mid t \in \specialtrans\}\cup\{p_{start}\}\cup\{p_{x_i}\mid i \in \{1,\dots,n\}\}\]

We call \textit{auxiliary token} any token on the \textit{auxiliary places} $p_{start}$,  $p_{t}^{fire}$, and $p_{t}^{rename}$. An \textit{auxiliary tuple} is a tuple with some auxiliary token.
The configuration encoding of a configuration $M$ of $W$ into a configuration $M'$ of $N$ extends $M$ by adding, for each tuple $m\in M$, a token in $p_{active}$, and adds a new token on $p_{start}$. 

\begin{definition}\label{def:cnupnstonupnsencoding}
Given a configuration $M=\fmset{m_1,\dots,m_\ell}$ of $W$, its encoding configuration in $W$ is $\xi(M) =\fmset{m_0',\dots,m_\ell'}$, where $m_0'=\fmset{p_{start}}$ and, for each $i\geq 1$,
\[m_i'(p)=\begin{cases}
m(p) & \text{if }m\in P\\
1 & \text{if }p=p_{active}\\
0 & \text{otherwise}
\end{cases}\]

\end{definition}
Note that $\xi$ is an injective function, i.e., it is a configuration encoding of $CG_W$ into $CG_N$.

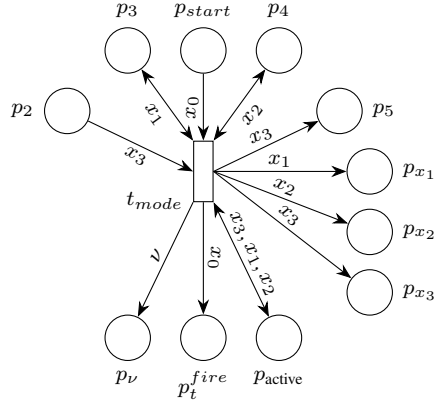
\begin{figure}[t]
        \centering
        \scalebox{1}{
        \begin{tikzpicture}

\node[place,label={[name=p2Lab]left:\scriptsize $p_2$}](p2)at (0.4,0.8){};
\node[place,label={[name=p3Lab]above:\scriptsize $p_3$}](p3)at (1.2,1.6){};
\node[place,label={[name=p4Lab]above:\scriptsize $p_4$}](p4)at (3.2,1.6){};
\node[place,label={[name=p5Lab]right:\scriptsize $p_5$}](p5)at (4.0,0.8){};
\node[place,label={[name=px1Lab]right:\scriptsize $p_{x_1}$}](px1)at (4.4,0){};
\node[place,label={[name=px2Lab]right:\scriptsize $p_{x_2}$}](px2)at (4.4,-.8){};
\node[place,label={[name=px3Lab]right:\scriptsize $p_{x_3}$}](px3)at (4.4,-1.6){};
\node[place,label={[name=pnLab]below:\scriptsize $p_\nu$}](pn)at (1.2,-2.2){};
\node[place,label={[name=paLab]below:\scriptsize $\pactive$}](pa)at (3.2,-2.2){};
\node[place,label={[name=pmLab]below:\scriptsize $\pinit{t}$}](pm)at (2.2,-2.2){};

\node[transvert,label={[name=tLab]below left:\scriptsize $t_{mode}$}] (t) at (2.2,0){};

\draw [->] (p2) to node [below, midway, sloped, yshift=2]  (TextNode1) {\scriptsize $x_3$} (t.west);

\draw [<->] (p3) to node [below, sloped]  (TextNode2) {\scriptsize $x_1$} (t.north west);

\draw [<->] (p4) to node [below, sloped]  (TextNode3) {\scriptsize $x_2$} (t.north east);

\draw [->] (t.south) to node [above, midway, sloped, yshift=-2]  (TextNode4) {\scriptsize $x_0$} (pm);

\draw [->] (t.south west) to node [above, midway, sloped, yshift=-2]  (TextNode5) {\scriptsize $\nu$} (pn);

\draw [<->] (t.south east) to node [above, midway, sloped, yshift=-2]  (TextNode6) {\scriptsize $x_3,x_1,x_2$} (pa);

\node[place,label={[name=pstartLab]above:\scriptsize $p_{start}$}](pstart)at (2.2,1.6){};
\draw [<-] (t.north) -- node [above, midway, sloped, yshift=-2] {\scriptsize $x_0$} (pstart) ;
\draw [->] (t.east) -- node [above, midway, sloped, yshift=-2] {\scriptsize $x_3$} (p5) ;
\draw [->] (t.east) -- node [above, midway, sloped, yshift=-2] {\scriptsize $x_1$} (px1) ;
\draw [->] (t.east) -- node [above, midway, sloped, yshift=-2] {\scriptsize $x_2$} (px2) ;
\draw [->] (t.east) -- node [above, midway, sloped, yshift=-2] {\scriptsize $x_3$} (px3) ;

\end{tikzpicture}
        }
        \caption{\(t_{mode}\) is fired to fixe the firing mode and to capture the standard part of $t$.} 
        \label{fig:p1}
\end{figure}
\begin{figure}[t]
        \centering
        \scalebox{1}{
        \begin{tikzpicture}

\node[place,label={[name=p1Lab]left:\scriptsize $p_1$}](p1)at (0,0){};

\node[place,label={[name=p6Lab]right:\scriptsize $p_6$}](p6)at (4.4,0){};

\node[place,label={[name=pmLab]left:\scriptsize $\pinit{t}$}](pm)at (0,-1.0){};
\node[place,label={[name=pfLab]right:\scriptsize $\pfire{t}$}](pf)at (4.4,-1){};

\node[place,label={[name=px1Lab]above:\scriptsize $p_{x_1}$}](px1)at (0,1.6){};
\node[place,label={[name=px2Lab]above:\scriptsize $p_{x_2}$}](px2)at (2.2,1.6){};

\node[transvert] (t) at (2.2,0){};
\node at (t) [below, xshift=15 ,yshift=-3] {\scriptsize $t_{fire}$};

\node[transvert,label={below:$t_{stop}$}] (t1) at (2.2,-1.0){};

\draw [->] (p1) to node [above, midway, sloped, yshift=2]  (TextNode1) {\scriptsize $x_1$} (t.west);

\draw [->] (t.east) to node [above, midway, sloped, yshift=-2]  (TextNode4) {\scriptsize $x_2$} (p6);

\draw [->] (t1.east) to node [above, midway, sloped]  (TextNode4) {\scriptsize $x_0$} (pf);

\draw [<->] (pm) to node [above, sloped]  (TextNode3) {\scriptsize $x_0$} (t.south west);

\draw [->] (pm) to node [below, sloped]  (TextNode3) {\scriptsize $x_0$} (t1.west);

\draw [<->] (t.north west) -- node [above, midway, sloped, yshift=-2] {\scriptsize $x_1$} (px1) ;
\draw [<->] (t.north) -- node [above, midway, sloped, yshift=-2] {\scriptsize $x_2$} (px2) ;

\end{tikzpicture}
        }
        \caption{\(t_{fire}\) can fire several times, non-deterministically stopped by the firing of \(t_{stop}\). This simulates the transfer on the channel one token per time.}
        \label{fig:p2}
\end{figure}
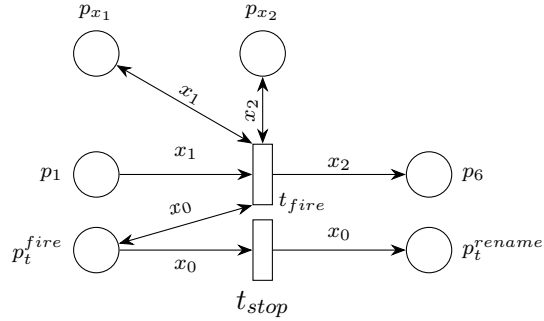
\begin{figure}[t]
        \centering
        \scalebox{1}{
        \begin{tikzpicture}

\node[place,label={[name=pnLab]below:\scriptsize $p_\nu$}](pn)at (1.1,-1.6){};

\node[place,label={[name=pfLab]below:\scriptsize $\pfire{t}$}](pf)at (3.3,-1.6){};
\node[place,label={[name=pmLab]right:\scriptsize $p$}](p')at (3.3,1.6){};

\node[place,label={[name=px1Lab]above:\scriptsize $p_{x_1}$}](px1)at (1.1,1.6){};

\node[transvert,label={right:\scriptsize $t_{rename}^p$}] (t) at (2.2,0){};

\draw [<->] (pn) to node [above, midway, sloped, yshift=2]  (TextNode1) {\scriptsize $x_2$} (t.south west);

\draw [<->] (pf) to node [above, sloped]  (TextNode3) {\scriptsize $x_0$} (t.south east);

\draw [->] (p') to node [below, sloped]  (TextNode3) {\scriptsize $x_1$} (t.north east);

\draw [<-] (p') to[in= 100, out=-160] node [above, sloped]  (TextNode3) {\scriptsize $x_2$} (t.north east);

\draw [<->] (t.north west) -- node [above, midway, sloped, yshift=-2] {\scriptsize $x_1$} (px1) ;

\end{tikzpicture}
        }
        \caption{The transitions \(t_{rename}^p\), for each $p\in P$, can fire several times concurrently, non-deterministically stopped by the firing of $t_{reset}$. At each firing, $t_{rename}^p$ renames the tokens, inside $p$, of the transferred tuple into the new name dictated by the channel, analogously to what is done by $t_{fire}$, but without the transfer. 
        }
        \label{fig:p3}
\end{figure}
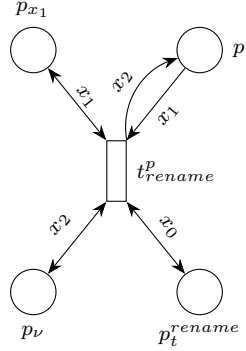
\begin{figure}[t]
        \centering
        \scalebox{1}{
        \begin{tikzpicture}

\node[place,label={[name=pnLab]below:\scriptsize $p_\nu$}](pn)at (1.2,-2.2){};
\node[place,label={[name=paLab]below:\scriptsize $\pactive$}](pa)at (3.2,-2.2){};
\node[place,label={[name=pfLab]below:\scriptsize $\pfire{t}$}](pf)at (2.2,-2.2){};

\node[transvert, label={left:$t_{reset}$}] (t) at (2.2,0){};

\draw [->] (pn) to node [below, midway, sloped, yshift=2]  (TextNode1) {\scriptsize $x_3$} (t.south west);

\draw [->] (t.south east) to node [above, midway, sloped, yshift=-2]  (TextNode4) {\scriptsize $x_3$} (pa);

\draw [->] (pf) to node [below, sloped]  (TextNode3) {\scriptsize $x_0$} (t.south);

\draw [<-] (t.south east) to[in =75, out =-30] node [above, midway, sloped, yshift=-2]  (TextNode4) {\scriptsize $x_1$} (pa);

\node[place,label={[name=px1Lab]above:\scriptsize $p_{x_1}$}](px1)at (1.1,1.6){};
\node[place,label={[name=px2Lab]above:\scriptsize $p_{x_2}$}](px2)at (2.2,1.6){};
\node[place,label={[name=px3Lab]above:\scriptsize $p_{x_3}$}](px3)at (3.3,1.6){};
\node[place,label={[name=px3Lab]above:\scriptsize $p_{start}$}](pstart)at (4.4,0){};
\draw [<-] (t.north west) -- node [above, midway, sloped, yshift=-2] {\scriptsize $x_1$} (px1) ;
\draw [<-] (t.north) -- node [above, midway, sloped, yshift=-2] {\scriptsize $x_2$} (px2) ;
\draw [<-] (t.north east) -- node [above, midway, sloped, yshift=-2] {\scriptsize $x_3$} (px3) ;
\draw [->] (t.east) -- node [above, midway, sloped, yshift=-2] {\scriptsize $x_0$} (pstart) ;

\end{tikzpicture}
        }
        \caption{\(t_{reset}\) disables the transitions \(t_{rename}^p\), deletes the tokens in the $p_{x_i}$, and moves the auxiliary token, preparing the net for a new phase. It also incorporates in $p_{active}$ the token in $p_{active}$. If some token to be transferred was left over, it is now not tracked by $p_{active}$ and returns a broken tuple.}
        \label{fig:p4}
\end{figure}
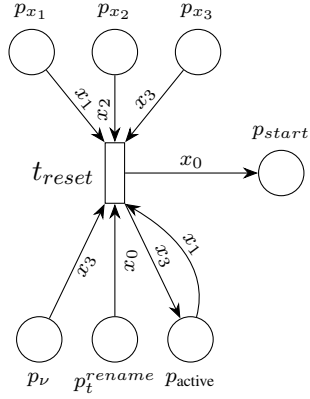

The set $T'$ of transitions of $N$ contains the standard transitions $t$ of $W$ along with their pre- and post-conditions, extended with the read (consumption and creation at the same time) of the auxiliary token from $p_{start}$ and of the variables $x\in\var(t)$ from $p_{active}$. Moreover, for each simplified special transition $t\in T_{spec}$, $T'$ captures $t$ via the sequential firing of several transitions. Specifically, for each $p\in P$, $T'$ contains the transitions $t_{mode}$, $t_{fire}$, $t_{stop}$, $t_{rename}^p$, and $t_{reset}$ depicted in Fig.\ref{fig:p1}, Fig.\ref{fig:p2}, Fig.\ref{fig:p3}, and Fig.\ref{fig:p4}. The transitions $t_{fire}$ and $t_{rename}^p$ capture the transfer with renaming performed by the channel by manipulating the tokens one by one through several firings, instead of all at once. However, since these transitions cannot detect the termination of this process, they may leave behind some leftovers, which are signaled by some tuple that do not place any token on $p_{active}$. We call such tuples \textit{broken}.
\begin{definition}
    A tuple $m$ over $P'$ is \textit{broken} if $m$ is also a tuple over $P$. A configuration $M$ of $W$ is \textit{broken} if there is some broken tuple $m\in\support{M}$.
\end{definition}

Because of the pre- and post-conditions of the transitions in $N$, starting from an encoding configuration, once a broken token is produced, it does not interact with the net anymore and cannot be repaired. Also, in each configuration there can only be a single auxiliary token, belonging to only one auxiliary tuple, denoted by $m_{aux}$. If $m_{aux}$ marks $p_{start}$, then $p_{\nu}$ is empty. Finally, a tuple can mark $p_{active}$ with at most one token.

Despite broken configurations, it turns out that $\xi$ is still a phase-encoding. This is because finite phases do not reach broken configurations.
\newcommand{\lemmaCGwCGNphaseencodingStatement}{
The function $\xi$ in Def.\ref{def:cnupnstonupnsencoding} is a phase encoding of $CG_W$ into $CG_N$.
}
\begin{lemma}\label{lemma:CGwCGNphaseencodingStatement}
\lemmaCGwCGNphaseencodingStatement
\end{lemma}

The proof is in App.\ref{app:phaseEncodingProofs}.

\subsection{Removing broken configurations}
Unfortunately, although $\xi$ is a phase-encoding, and $CG_W$ is finitary, $\xi$ is not uniform nor phase-finite and, thus, we cannot directly apply the lemmas in Sec.\ref{sec:phaseencoding}. However, these properties are regained if we project out broken configurations, obtaining an intermediate qonCG.
\begin{definition}
    The qonCG $CG^W_N$ is the qonCG of $N$ where all broken configurations $K$ and, for each configuration $H$, steps $K\rightarrow H$ and $H\rightarrow K$ have been dropped.
\end{definition}
Note that $CG^W_N$ corresponds to a constrained version of $CG_N$, in which no cheating can happen, i.e., no broken configuration can be reached.
Since encoding configurations are not broken, $\xi$ remains an encoding of the qonCG of $CG_W$ into $CG^W_N$. Moreover, since broken configurations play no role in making $\xi$ a phase encoding of $CG_W$ into $CG_N$, $\xi$ is also a phase encoding of $CG_W$ into $CG^W_N$
\begin{corollary}
    The function $\xi$ in Def.\ref{def:cnupnstonupnsencoding} is a phase-encoding of $CG_W$ into $CG^W_N$.
\end{corollary}

Moreover, in this setting, $\xi$ becomes perfect.
\begin{lemma}\label{lemma:perfectChanneltoNu}
    The function $\xi$ in Def.\ref{def:cnupnstonupnsencoding} is an $f$-perfect phase encoding of $CG_W$ into $CG^W_N$.
\end{lemma}
\begin{proof}
    \begin{description}
        \item[Embedding] Clearly, $\xi$ is an embedding.
        \item[Uniform] If $H$ is reachable from an encoding configuration and $H\succeq \xi(K)$, then the auxiliary token of $H$ is on $p_{start}$. The other tokens of $H$ can only be on places in $P$ or on $p_{active}$ and each tuple can place at most one token on the latter.
        
        Since $H$ is not broken (there are no broken configurations in $CG^W_N$), for each $m\in\support{H}\setminus\{m_{aux}\}$ and $p\in P$, if $m(p)\neq 0$ then $m(p_{active})=1$. Thus, $H$ is an encoding configuration.
        
        \item[Phase-finite] For each phase $\pi$, either $\pi$ consists in the firing of a single standard \nPN transition of $T$, or, because of the flow of auxiliary tokens moved by the transitions $t_{mode}$, $t_{fire}$, $t_{stop}$, $t_{rename}^p$, and $t_{reset}$, $\pi$ consists at least in the firing from the encoding configuration $\start(\pi)$ of a prefix of a run of the form $\sigma:\start(\pi)\rightarrow^{t_{mode}}K_1\rightarrow^{\sigma_1}K_2\rightarrow^{t_{stop}}K_3\rightarrow^{\sigma_2} K_4 \rightarrow^{t_{reset}} K_5$, where $\sigma_1$ contains only finitely many firings of $t_{fire}$ and $\sigma_2$ contains only finitely many firings of the transitions $t_{rename}^p$, for some special \GnPN $t\in T$. In fact, note that the transitions $t_{fire}$ and $t^p_{rename}$ cannot fire infinitely often: each of their firing removes one token from the tuple marking $p_{x_1}$, while no other transition in the run add tokens back in this tuple. Thus, $\sigma$ is finite. Moreover, no configuration before $K_5$ is an encoding configuration, so that $\pi$ contains at least $\sigma$. However, $K_5$ is obtained after a full simulation cycle of the special transition $t$. Thus, it is either an encoding configuration or a broken configuration. Since $C^W_N$ does not have any broken configuration, $K_5$ has to be the first encoding configuration in $\pi$ after $K_1$. Hence, $\pi$ coincides with $\sigma$, which is finite. 
        
        \item[$f$-bounded] Each firing sequence in $CG^W_N$ outgoing from an encoding configuration and not reaching an encoding configuration has a form as $\sigma$ in the previous paragraph. Only $t_{mode}$ creates new tokens, while the other transitions simply rearrange them. Since the norm of $CG^W_N$ is the same as that in $CG_N$, i.e., it is the token counting function, $f$-boundedness holds if $f(n)=n+\max_{t\in T}\{\norm{\var(t)},post(t),3\}$.
    \end{description}
\end{proof}

Consequently, $\xi$ maps instances of the problems over $W$ into equivalent instances over $CG^W_N$.

\subsection{Reintroducing Broken configurations}
We can now reintroduce the broken configurations obtaining equivalent instances of the various problems over $CG^W_N$. This can be shown by exploiting compatibility\footnote{If $K\geq H\rightarrow^* H'$, then there is some $K'$ such that $K\rightarrow^*K'\geq H$.} of \nPNs and the following lemmas.
\begin{lemma}
    If $\sigma:\xi(K)\rightarrow^*H\geq\xi(K')$ and $H$ is the first broken configuration in $\sigma$, then there is some $K''\geq \xi(K')$ and a run $\sigma':\xi(K)\rightarrow^*\xi(K'')\geq \xi(K')$.
\end{lemma}
\begin{proof}
By the distribution of black tokens in $N$, $\sigma$ can reach a broken configuration only if 
$\sigma:\start(\pi)\rightarrow^{t_{mode}}K_1\rightarrow^{\sigma_1}K_2\rightarrow^{t_{stop}}K_3\rightarrow^{\sigma_2} K_4 \rightarrow^{t_{reset}} K_5$ where $\sigma_1$ contains only finitely many firings of $t_{fire}$ and $\sigma_2$ contains only finitely many firings of the transitions $t_{reset}^p$, but $t_{fire}$ or some $t_{reset}^p$ left some left-overs. Thus, we can extend 
$\sigma$ by adding as many firings as possible of such transitions, resulting in a run of $CG^W_N$ that reaches an encoding configuration $\xi(K'')$. In fact, note that adding these firings do not disable following firings in $\sigma$. Moreover, because of the shape of $t_{fire}$ and $t_{reset}^p$, $\xi(K'')$ retains all non-broken tuples of $H$, i.e., $\xi(K'')$ covers all encoding configurations already covered by $H$, i.e., $\xi(K'')\geq \xi(K')$.
\end{proof}
The next lemma is an immediate consequence of compatibility and the fact that no transition of $W$ can instantiate variables to broken tuples.
\begin{lemma}
Let $H_i$, for $i\in\mathbb{N}$, be broken configurations and $K_0$ be the encoding configuration obtained from $H_0$ by removing the broken tuples. If there is an infinite run $H_0\rightarrow^{*}H_1\rightarrow^{*} \dots$, then there is also a run $K_0\rightarrow^{*}K_1\rightarrow^{*} \dots$ where each $K_i$ is an encoding configuration.
\end{lemma}
It is now easy to see that any yes-instance $\xi(I)$ of a problem over $CG^W_N$, obtained from an instance $I$ over $CG_W$, is equivalent to the instance $\xi(I)$ over $CG_N$. By putting together this fact with Lemma\ref{lemma:perfectChanneltoNu} and the lemmas in Sec.\ref{sec:phaseencoding} we obtain the following result.
\begin{theorem}\label{thm:brokenChanneltoNu}
    $I$ is a yes instance of \GnPN-coverability (respectively \GnPN-termination, \GnPN-boundedness) if and only if it $\xi(I)$ is a yes instance of the \nPN-coverability (\nPN-termination, \nPN-boundedness).
\end{theorem}
Since $N$ and the instances $\xi(I)$ can be built from $W$ and $I$ in polynomial time, we have a reduction of coverability, termination, and boundedness from \GnPNs to \nPNs. 

Overall, for each of these problems $\Pi$, we have a cycle of reductions from \nPNs, to cEOSs, to \GnPNs and back to \nPNs. Thus, the complexity of $\Pi$ is constant over all these variants of PNs. This allows us to obtain the following theorem.
\begin{theorem}
    The complexity of cEOS-coverability, cEOS-termination, and cEOS-boundedness is the same as \nPN-coverability ($F_{\omega2}$-complete), \nPN-termination (non-primitive recursive), and \nPN-boundedness (non-primitive recursive).
\end{theorem}

\section{Conclusions}\label{sec:conclusions}

We have inter-reduced cEOS-coverability, cEOS-termination, and cEOS-boundedness to their corresponding problems on \nPNs. This yields an $F_{\omega2}$-completeness for cEOS-coverability and non-primitive recursive lower-bounds for cEOS-termination and cEOS-boundedness. Interestingly, only one of our reduction performs cheating steps, namely the one reducing \GnPNs into \nPNs. This allows us to interpret the nested nets paradigm as cEOS verification in the hierarchy of data nets, specifically in between \nPNs and Unordered Data Nets, thus bridging two apparently orthogonal approaches of PN extensions, at least from the perspective of verification.
In fact, one could use the same reductions to reduce other problems, e.g., which it is already known to be undecidable on both \nPNs and cEOSs. Instead, the reduction from \GnPNs to \nPNs performs cheating steps, in order to capture whole place transfers with standard transitions. Thus, reductions can be obtained only for a restricted class of \textit{coverability-like} problems. However, further results of this kind, would enable a deeper understanding of the computational power of PN nesting. 

In this work we have studied object net systems restricted to two levels. In existing literature~\cite{DBLP:conf/ac/Valk03}, there is only a formal definition of EOSs but there is no formal definition of object systems with nesting depth beyond $2$. Intuitively, the restriction on the nesting depth should in turn restrict the complexity of the verification problems, when compared to object systems with several layers of nesting. 
We hypothesize that given that the nesting levels are finite, we can generalize the reduction from CEOSs to \GnPNs by adding a fixed hierarchy of places in the system nets to again capture nested object via identifiers. This remains an interesting direction to be explored in future works.

\bibliographystyle{fundam}
\bibliography{citations}

\appendix
\section{Phase Encodings Invariance}\label{app:phaseProofs}
We recall the lemmas from Sec.~\ref{sec:phaseencoding} and provide their proofs below:

\begin{lemmanum}{\ref{coverLemmaLabel}}
    \coverLemmaStatement
\end{lemmanum}
\begin{proof}
    If $I=(\C_1,K_0,H)$ is a yes-instance of $F_1$-coverability, then $K_0\rightarrow^\ast K_n \geq H$, for some $n\in \mathbb{N}$ and $K_n\in C_1$. Since $\xi$ is an embedding, $\xi(K_n)\geq \xi(H)$. Moreover, since $\xi$ is a phase-encoding, for each step $K_i\rightarrow K_{i+1}$ in $K_0\rightarrow^\ast K_n$ there is a corresponding phase $P_i=\xi(K_i)\rightarrow^\ast \xi(K_{i+1})$ in $\C_2$. These phases are consistent. Thus, we have a run $\xi(K_0)\rightarrow^* \xi(K_n)$, made of several phases. Hence $\xi(I)$ is a yes instance of $F_2$-coverability.

    Vice-versa, if $\xi(I)=(\C_2,\xi(K_0),\xi(H))$ is a yes-instance of $F_2$-coverability, then there is a path $\xi(K_0)\rightarrow^\ast H' \geq \xi(H)$. Since $\xi$ is uniform, then there is some $K_n$ such that $H'=\xi(K_n)$. Since $\xi$ is an embedding, we have $K_n\geq H$. Moreover, the finite sequence $\xi(K_0)\rightarrow^\ast \xi(K_n)$ is either a single phase or a finite composition of consistent phases $P_1,\dots,P_n$. Since $\xi$ is injective, we have $K_0=\xi^{-1}(\start(P_1))\rightarrow \xi^{-1}(\start(P_2))\rightarrow\dots\rightarrow \xi^{-1}(\start(P_n))\rightarrow \xi^{-1}(\finish(P_n))=K_n$. Summarizing, $K_0\rightarrow^\ast K_n\geq H$, i.e., $I$ is a yes-instance of $F_1$-coverability.
\end{proof}

\begin{lemmanum}{\ref{terminationLemmaLabel}}
    \terminationLemmaStatement
\end{lemmanum}
\begin{proof}
    If $I=(\C_1,K_0)$ is a yes-instance of $F_1$-termination, then there is an infinite run $K_0\rightarrow K_1\rightarrow \dots$ in $\C_1$. Thus, there is an infinite run $\xi(K_0)\rightarrow^\ast \xi(K_1)\rightarrow^\ast \dots$, where, for each $i\in\mathbb{N}$, $\xi(K_i)\rightarrow^* \xi(K_{i+1})$ is a phase. Thus $\xi(I)$ is a yes-instance of $F_2$-termination.

    If $\xi(I)=(\C_2,\xi(K_0))$ is a yes-instance of $F_2$-termination, then there is an infinite run $\xi(K_0)=H_0\rightarrow H_1\rightarrow \dots$ in $\C_2$. Since $\xi$ is phase-finite, each phase is finite. However, the run is infinite. Thus, it contains infinitely many encoding configurations (otherwise it would be an infinite phase), corresponding to infinitely many finite phases. So we can partition the run into an infinite sequence of consistent phases $P_0, P_1, \dots$ and obtain the run $K_0=\xi^{-1}(\start(P_1))\rightarrow \xi^{-1}(\start(P_2))\rightarrow\dots$ in $\C_1$. This run is infinite and starts from $K_0$. Thus, $I$ is a yes-instance of $F_1$-termination.
\end{proof}

\begin{lemmanum}{\ref{boundedLemmaLabel}}
    \boundedLemmaStatement
\end{lemmanum}

\begin{proof}
        If $I=(\C_1,K_0)$ is a yes-instance of $F_1$-boundedness, then the set of encoding configurations reachable from $\xi(K_0)$ in $\C_2$ is finite. Thus, the norms of encoding reachable configurations in $\C_2$ from $\xi(K_0)$ is bounded by some fixed number $n_{K_0}$.
        Let $H$ be a non-encoding configuration of $\C_2$ reachable from $\xi(K_0)$, via some run $\sigma:\xi(K_0)\rightarrow^* H$. Let $K$ be the last encoding configuration in $\sigma$. 
        Thus, $\norm{K}\leq n_{K_0}$ and $H$ is reachable from $K$ without stepping on encoding configurations. 
        Since $\xi$ is $f$-bounded, $\norm{H}\leq f(\norm{K})\leq f(n_{K_0})$. 
        Thus, $H\in\bigcup_{i=0}^{f(n_{K_0})} \norm{\bullet}_2^{-1}(i)$, which is a finite union of finite sets, since $\C_2$ is finitary, i.e., it is finite. Thus, both the sets of encoding and of non-encoding configurations reachable from $\xi(K_0)$ are finite. 
        Summarizing, the set of configurations reachable from $\xi(K_0)$ is finite, i.e., $\xi(I)$ is a yes-instance of $F_2$-boundedness.

        Vice-versa, let $\xi(I)=(\C_2,\xi(K_0))$ be a yes-instance of $F_2$-boundedness. Then, the set $\R$ of reachable encoding configurations in $\C_2$ from $\xi(K_0)$ is finite. Since $\xi$ is injective, also $\xi^{-1}(\R)$ is finite. Moreover, for each configuration $K$ reachable from $K_0$ in $\C_1$, we have $K\in \xi^{-1}(\R)$. Thus, the set of configurations reachable from $K_0$ in $\C_1$ is a subset of the finite set $\xi^{-1}(\R)$. Thus, $I$ is a yes-instance of $F_1$-boundedness.
\end{proof}
 
\section{Phase-encoding Proofs}\label{app:phaseEncodingProofs}
This appendix contains the proofs of the theorems stating that the various encodings we exploited are actually phase-encodings.

\begin{lemmanum}{\ref{lemma:nuToCEOSphaseEncoding}}
    \nuToCEOSphaseEncodingStatement
\end{lemmanum}
\begin{proof}
First, we show that for each step in $\D$ there is a corresponding phase in $\os$.
Let $M\rightarrow^{t,e}M'$ be a step of $\D$, for some configurations $M,M'\in C_{\D}$, transition $t\in T$ and \nPN mode $e$.
We assume $\var(t)=\{x_1,\dots,x_n,\nu\}$; in case $\nu\notin\var(t)$, the argument is analogous. Note that, by \nPN semantics, $M=\fmset{m_1,\dots,m_n}+M''$ for some tuples $m_i$ and configuration $M''$, and that, for each $x_i\in\var(t)$, the tuple $e(x_i)$ enables $t_{x_i}$. Thus, we have 

\[\xi(M)=\sum_{i=1}^n\tup{sim,m_i}+\sum_{m\in M''}\tup{sim,m}+\tup{selectTran,\varepsilon}\]

Consequently, the system autonomous event $t^{select}_{x_1}$ is enabled under a cEOS mode that moves an object with internal marking $m_i=e(x_i)$ from $sim$ to $t^{selected}_{x_1}$ and $\blacksquare$ to $select^{t}_{x_2}$. Now, only the event $t^{select}_{x_2}$ is enabled. We can iterate this argument so as to fire, in sequence, all events $t^{select}_{x_i}$ for $x_i\in\var(t)$ and $t^{select}_\nu$, obtaining the configuration

\[\sum_{m\in M''}\tup{sim,m}+\sum_{i=1}^n(\tup{t^{selected}_{x_i},m_i}+\tup{t^{run}_{x_i},\blacksquare})+\tup{t^{run}_\nu,\blacksquare}\]

Now, because of the distribution of the $\blacksquare$ tokens and because each $M_i=e(x_i)$ enables $t_{x_i}$, each event $t^{fire}_{x_i}\tup{t_{x_i}}$ is enabled with a mode $(\lambda,\rho)$ where $\lambda=\tup{t^{selected}_{x_i},m_i}+\tup{t^{run}_{x_i},\varepsilon}$. These events fire concurrently; no other event can get enabled before the last $t^{fire}_{x_i}\tup{t_{x_i}}$ has fired because of the distribution of the $\blacksquare$ tokens. Overall this concurrent firing reaches the configuration 
\[\sum_{m\in M''}\tup{sim,m}+
\sum_{i=1}^n(\tup{sim,m_i - pre(t_{x_i})+post(t_{x_i})}
-\tup{selectTran,\blacksquare}+\tup{t^{done},(n+1)t^{report}}\]

 where only $t^{done}$ is enabled. In turn, the firing of the event $t^{done}$ reaches
\[ \sum_{m\in M''}\tup{sim,m}+\xi(M')=\xi (M')\]
Thus, we have a finite phase $\xi(M)\rightarrow^*\xi(M')$.

We now show that for each phase in $\os$ there is a corresponding step in $\D$. If there is a phase $\pi$ such that $\xi({M})\rightarrow\*\xi(M'')$, then we can organize $\pi$ into three blocks: 
\begin{inparaenum}[\itshape (1)]
\item A prefix run that amounts to the firing of the $t^\text{select}_x$ transitions, reaching a configuration 
\[\xi(M)-\sum_{i=1}^{n} \tup{sim, m_i} -\tup{selectTran,\varepsilon}+ \sum_{i=1}^{n} \tup{t^{selected}_{x_i}, m_i}+ \sum_{i=1}^{n} \tup{t^{run}_{x_i}, m_i} \]
where $n=|\var(t)|$.
\item An intermediate run $\sigma$ that amounts to the firing of the transitions $t^{\text{fire}}_{x}$. Since in this run each $t^{fire}_{x_i}$ fires, we have that each $m_i$ enables $t_i$. Assuming $\nu\in\var(t)$, it reaches a configuration

\begin{align*}
    \xi(M)-\sum_{i=1}^{n} \tup{sim, m_i} -\tup{selectTran,\varepsilon}+ 
    \sum_{i=1}^{n} \tup{sim, m_i-pre(t_{x_i})+post(t_{x_i})} + \\
    \tup{sim,post(t_{\nu})}+\tup{t^{report},(n+1)\varepsilon}
\end{align*}

If $\nu\not\in\var(t)$, the last addend has to be dropped. 
\item the firing of $t^{\text{report}}$, for $x\in\var(t)$.
\end{inparaenum}
Overall, taking into account the definition of $\xi$, we have
\begin{align*}
 \xi(M'')=& \xi(M)-\sum_{i=1}^{n} \tup{sim, m_i} + \sum_{i=1}^{n} \tup{sim, m_i-pre(t_{x_i})+post(t_{x_i})} + \tup{sim,post(t_{\nu})}\\
 =&\xi(M-\sum_{i=1}^{n} m_i+ m_i-pre(t_{x_i})+post(t_{x_i}))
\end{align*}
Thus, by injectivity of $\xi$, $M\rightarrow^{t,e}M''$ where $e$ is the \nPN mode such that $e(x_i)=m_i$.
\end{proof}

\begin{lemmanum}{\ref{thm:phase_encoding_eos_to_cnpn}}
\phaseencodingeostocnpnstatement
\end{lemmanum}    
\begin{proof}
Given an \(\os = (\hat{N},\N,d,\Theta)\), let the event $e=\tup{\tau,\theta}$ is enabled on the configuration $M=\fmset{\tup{\hat{p}_1,m_1}, \tup{\hat{p}_2,m_2}, \cdots, \tup{\hat{p}_{\abs{M}},m_{\abs{M}}}}$ with mode $(\lambda,\rho)$, and \(M'=\fmset{\tup{\hat{p}'_1,m'_1}, \tup{\hat{p}'_2,m'_2}, \cdots, \tup{\hat{p}'_{\abs{M'}},m'_{\abs{M'}}}}\) where $\lambda=\fmset{\tup{\hat{p}_{i_1},m_{i_1}}, \tup{\hat{p}_{i_2},m_{i_2}}, \cdots, \tup{\hat{p}_{i_{\abs{\lambda}}},m_{i_{\abs{\lambda}}}}}$ and $\rho=\fmset{\tup{\hat{p}_{j_1}',m'_{j_1}}, \tup{\hat{p}_{j_2}',m_{j_2}'}, \cdots, \tup{\hat{p}_{j_{\abs{\rho}}'},m_{j_{\abs{\rho}}}'}}$ where each \(i_k \in [\abs{M}]\) and \(j_{k'} \in [\abs{M'}]\). Recall that this proof is about the construction presented in Sec.~\ref{sec:fromEOS}, which in turn assumes that $\lambda$ and $\rho$ deal with a single system net type, i.e., $d(\hat{p}_{i_1}) = \cdots = d(\hat{p}_{i_{\abs{\lambda}}})= d(\hat{p}_{j_1})=\dots=d(\hat{p}'_{j_{\abs{\rho}}})=N$, for some $N\in\N$.

We define an injection \(H: [\abs{\lambda}] \rightarrow \Pi_1(\lambda) \times [\abs{\lambda}]\) that maps each pair of lambda to its system place and position in $\lambda$ when restricted only to that place, i.e., \(H(k)=(\hat{p}_{i_k},\ell)\) where \(\ell =|\{j \mid j\in\{1,\dots,k\}, \hat{p}_{i_j}=\hat{p}_{i_k}\}|\).

The transition $\tau_e^{merge}$ is enabled with mode $\varepsilon_{mer}$ such that :
\begin{itemize}
    \item $\varepsilon_{mer}(x_{cs})=\delta_{p^{init}}$. From now on, we will always associate the variable \(x_{cs}\) to the control sequence tuple.
    \item $\varepsilon_{mer}(x^{\hat{p}}_\ell)=\K(m_k)$ such that $H(k)=(\hat{p},\ell)$.
\end{itemize}
Note that we can write \(M = M'' + \lambda\) and \(M' = M'' + \rho\), for some \(M''\). 
After firing \(\tau_e^{merge}\) using the mode \(\varepsilon_{mer}\) on the configuration \(\xi(M) = \K(M'') + \K(\lambda) + \fmset{\delta_{p^{init}}}\), 
we get the resultant marking 
\[\hat{M}_{merged} = \K(M'') + \fmset{m_{merged}} + \fmset{\delta_{p^{merged}_e}}\]
where \(m_{merged} = \sum\limits_{\substack{<\hat{p},m> \in \lambda}} \sum\limits_{p \in P_{N}} m(p)\delta_{p^{mer}} + \delta_{\Id^m_N} = \K_{mer}(\Pi^2_{N}(\lambda))\). 
This is the configuration obtained after executing the module \textbf{$e$-merging}. 

Now, \(\delta_{p^{merged}_e}\) enables the execution of the \textbf{\(e\)-updating} module. After firing transitions \(\tau^{init}_e\), all the \(t_e\) transitions, and \(\tau^{fin}_e\) sequentially in the module \textbf{\(e\)-updating},  with a unique enabling mode (this mode exists since $(\lambda,\rho)$ is an enabling mode for \(e\)), we get the resultant marking 
\[\hat{M}_{updated} = \K(M'') + \fmset{m^0_{updated}} + \fmset{\delta_{p_e^{fin}}}\]
where 
\(m^0_{updated} = 
\sum\limits_{\substack{\tup{\hat{p}',m'} \in \rho}} \sum\limits_{p \in P_{N}}
m'(p)\delta_{p^{upd}} + \fmset{\delta_{\Id^u_{N}}}= \bar{\K}_{upd}(N,\Pi^2_{N}(\lambda) - \prefun_{N}(\theta(N)) + 
\postfun_{N}(\theta(N))) = 
\bar{\K}_{upd}(N,\Pi^2_{N}(\rho))\) by EOS semantics.

The token in $\delta_{p^{fin}_e}$ ensures that the control flows to the \textbf{$e$-distribution} module, consisting of sub-modules: \textbf{$e$-id-creation}, \textbf{$e$-move(i)} for each \(\tup{\hat{p}'_{i},n'_{i}} \in \rho\) and \textbf{$e$-transfer}, in that order). 

In the \textbf{$e$-id-creation} sub-module, on firing $t_e^{id}$,
a token is placed in $p^{new}_e$ for each new tuple generated, against each object net created by the event \(e\). Hence, we end up with the following resultant marking: 
\[\hat{M}_{created} = \K(M'') + \fmset{m^0_{updated}} + \fmset{c_1,c_2,\cdots,c_{\abs{\rho}}} + \fmset{\delta_{p^{move(1)}_e}}\]
where \(c_i = \delta_{p^{new}_e}\) for each \(i \in [\abs{\rho}]\).

Now, we can execute\textbf{ \(e\)-move($i$)} and \textbf{$e$-transfer} modules sequentially according to each \(r_i = \tup{\hat{p}'_{j_i},m'_{j_i}} \in \rho\). 
Specifically, for each $p\in P_N$, we fire $m'_{j_i}(p)$ times the transition of $e$-\textbf{move}-$(i)$ that consumes from $p^{upd}$.
After firing all the transitions in an \textbf{\(e\)-move($i+1$)} module for \(i \in \{0,\cdots,\abs{\rho}-2\}\), with the mode that instantiates \(x_{N}\) with \(m^i_{updated}\) and \(x_{new}\) with \(c_{i+1}\), we get,
\begin{align*}
    \hat{M}^{i+1}_{moved} =& \hat{M}^{i}_{moved} - \fmset{m^{i}_{updated}} + \fmset{m^{i+1}_{updated}} \\ 
    & - \fmset{c_{i+1}} + \K(r_{i+1})  - \fmset{\delta_{p^{move(i+1)}_e}} + \fmset{\delta_{p^{move(i+2)}_e}} 
\end{align*}
where $\hat{M}^{0}_{moved} = \hat{M}_{created}$, \(m^{i+1}_{updated} = \fmset{m^{i}_{updated} - \K_{upd}(r_{i+1})+\delta_{\Id^u_N}}\) 
and, with a slight abuse of notation, we define \(p^{move(\abs{\rho})}_e = p^{transfer}_e\).
Note that
\(\hat{M}_{moved}^i = \K(M'') + \fmset{m^i_{updated}} + \fmset{\K(r_1),\cdots,\K(r_i)} + \fmset{c_{i+1},\cdots,c_{\abs{\rho}}} + \fmset{\delta_{p^{move(i+1)}_e}}\)

Now, after executing module \textbf{\(e\)-move($\abs{\rho}-1$)} (or module $e$\textbf{-id-creation} if $\abs{\rho}=1$), \(\tau^{transfer}_e\) is enabled. After firing \(\tau^{transfer}_e\), we get the resultant marking:
\begin{align*}
    \hat{M}_{trans} =& \hat{M}^{\abs{\rho}-1}_{moved} - \fmset{m^{\abs{\rho}-1}_{updated}} + \K(r_{\abs{\rho}}) -\fmset{\delta_{p^{transfer}_e}} + \fmset{\delta_{p_{init}}} \\
=&\K(M'') + \fmset{\K(r_1),\cdots, \K(r_{\abs{\rho}})} + \fmset{\delta_{p^{init}}} \\
=&\K(M'') + \K(\rho) + \fmset{\delta_{p^{init}}} = \xi(M')
\end{align*}
where the first equation follows from $m^{\abs{\rho}-1}_{updated} = m^0_{updated} - \sum\limits_{i \in \abs{\rho}-1}(\K_{upd}(r_{i})) = \K_{upd}(N,\Pi^2_{N}(\rho)) - \sum\limits_{i \in \abs{\rho}-1}(\K_{upd}(r_{i})) = \K_{upd}(r_{\abs{\rho}})$. Thus, $\xi(M)\rightarrow^\ast \xi(M')$.

The other direction of the proof is analogous.
\end{proof}

\begin{lemmanum}{\ref{lemma:CGwCGNphaseencodingStatement}}
\lemmaCGwCGNphaseencodingStatement
\end{lemmanum}

\begin{proof}
Given a configuration \(H\) from \(CG_W\), if \(H \rightarrow^{t,e} K\) where \(H = M + \sum\limits_{x \in \X(t)} \fmset{m_{e(x)}}\) and \(t \in T_{std}\) then, 
\[K = M + \out + \sum\limits_{x \in \X(t)} \fmset{m'_{e(x)}}\] where for each \(x \in \X(t)\), \(m'_{e(x)} = m_{e(x)} - F_{x}(P,t) + F_{x}(t,P)\).
Recall that for \(t \in T_{std}\), we have \(t' \in T'\) such that for all \(p \in P\), \(F(t,p) = F'(t',p)\), \(F(p,t) = F'(p,t')\), \(F'(\pactive,t') = \X(t)\), \(F'(t',\pactive) = \X(t) \cup \Upsilon(t)\) and \(F'(p_{start},t') = F'(t',p_{start}) = 1\). 

Now, by \cref{def:cnupnstonupnsencoding}, we have \(\xi(H) = \xi(M) + \fmset{p_{start}} + \sum\limits_{x \in \X(t)} \xi(m_{e'(x)})\) where for \(x \in \X(t)\), the mode \(e'(x) = e(x)\) and \(e'(x_0) = \fmset{p_{start}}\). Now, it is easy to see that \(\xi(H) \rightarrow^{t',e'} K'\) where \[K' = \xi(M) + \xi(\out) + \fmset{p_{start}} + \sum\limits_{x \in \X(t)} \xi(m'_{e'(x)}) = \xi(K)\].

Symmetrically, observe that if \(\xi(H) \rightarrow^{t',e'} \xi(K)\), then, \(H \rightarrow^{t,e} K\). 

Now, we will tackle the case when \(t \in T_{spl}\) as depicted in \cref{fig:selectiveTransferPre}. Let \(H = M + \fmset{m_{e(x_1)},m_{e(x_2)},m_{e(x_3)}}\) be a configuration in \(CG_W\). Let \(H \rightarrow^{t,e} K\) where \[K = M +\fmset{m'_{e(x_1)},m'_{e(x_2)},m'_{e(x_3)}}\] such that 
\(m'_{e(x_1)} = \begin{cases}
    0 &\text{if }p=p_1\\
    m_{e(x_1)[p]} &\text{otherwise}
\end{cases}\),
\(m'_{e(x_2)} = \begin{cases}
    m_{e(x_2)[p]} + m_{e(x_1)[p]} &\text{if }p=p_2\\
    m_{e(x_2)[p]} &\text{otherwise}
\end{cases}\) 
and \(m_{e(x_3)}' = m_{e(x_3)} - F_{x_3}(P,t) + F_{x_3}(t,P)\). 
By \cref{def:cnupnstonupnsencoding}, we have \[\xi(H) = \xi(M) + \fmset{p_{start}} + \fmset{\xi(m_{e_{(x_1)}}),\xi(m_{e{(x_2)}}),\xi(m_{e_{(x_3)}})}\].
Then we show that \(\xi(H) \rightarrow^\ast \xi(K)\) where
\[\xi(K) = \xi(M) + \fmset{p_{start}} + \fmset{\xi(m'_{e(x_1)}),\xi(m'_{e(x_2)}),\xi(m'_{e(x_3)})}\].

We, in fact, show that there is a finite sequence of firings \(\sigma\): \(\xi(H) \rightarrow^{t_{mode},e_1} K_1 \rightarrow^{\sigma_1} K_2 \rightarrow^{t_{stop},e_2} K_3 \rightarrow^{\sigma_2} K_4 \rightarrow^{t_{reset},e_3} K_5\) below:

\begin{itemize}
    \item \(\xi(H) \rightarrow^{t_{mode},e_1} K_1\) where the transition \(t_{mode}\) is specified in \cref{fig:p1}. The mode \(e_1\) is defined as \(e_1(x_i) = \xi(m_{e(x_i)})\) for all \(i \in \{1,2,3\}\) and \(x_0\) selects \(\fmset{p_{start}}\). We get \[K_1 = \xi(M) + \fmset{p_t^{fire}} + \fmset{p_\nu} + \fmset{m_{e(x_1)}^1, m_{e(x_2)}^1, m_{e(x_3)}^1}\] where for all \(i \in \{1,2,3\}\), 
    
    \(m_{e(x_i)}^1[p] = \begin{cases}
        1 &{if } p= p_{x_i}\\
        m_{e(x_i)}^1[p]-1 &\text{if }p= p_2 \text{ and } i=3\\
        m_{e(x_i)}^1[p]+1 &\text{if }p= p_5 \text{ and } i=3\\
        m_{e(x_i)}^1[p] &\text{otherwise}
    \end{cases}\)
    \item \(K_1 \rightarrow^{\sigma_1} K_2 \rightarrow^{t_{stop},e_2} K_3\) where \(\sigma_1\) is a sequence of firings of the transition \(t_{fire}\) followed by the transition \(t_{stop}\) (see Fig.~\ref{fig:p2}). Note that \(t_{fire}\) can be fired at most \(m_{x_1}[p_1]\) times. The mode is fixed thanks to the places \(p_{x_i}\) for \(i \in \{1,2\}\), where \(e(x_i) = m^1_{e(x_i)}\) and \(x_0\) selects \(\fmset{p_{p^{fired}_t}}\). If fired exactly \(m_{x_1}[p_1]\) times, followed by the transition \(t_{stop}\) enforced with the same mode defined for \(t_{fire}\), we get the resultant configuration \[K_3 = \xi(M) + \fmset{p_t^{rename}} + \fmset{p_\nu} + \fmset{m_{e(x_1)}^2, m_{e(x_2)}^2, m_{e(x_3)}^2}\] where 
    \(m_{e(x_1)}^2[p] = \begin{cases}
        0 &\text{if }p=p_1\\
        m_{e(x_1)}^1[p] &\text{otherwise}
    \end{cases}\), 
    \(m_{e(x_2)}^2[p] = \begin{cases}
        m_{e(x_2)}^1[p]+m_{e(x_1)}^1[p] &\text{if }p=p_1\\
        m_{e(x_1)}^1[p] &\text{otherwise}
    \end{cases}\) and \(m_{e(x_3)}^2 = m_{e(x_3)}^1\).
    \item \(K_3 \rightarrow^{\sigma_2} K_4\) where \(\sigma_2\) is a sequence of firings of transitions \(t^p_{rename}\) for each \(p \in P\) (see Fig~\ref{fig:p3}). Note that the number of firings in \(\sigma_2\) is at most \(|m'_{e(x_1)}|\) times. The mode is again fixed by the places \(p_{x_1}\) that forces \(x_1\) to select \(m^2_{e(x_1)}\) and \(p_\nu\) that forces \(x_2\) to select the tuple \(\fmset{p_\nu}\). If fired exactly \(|m'_{e(x_1)}|\) times, we get the configuration \[K_4 = K_3 + \fmset{}\].
    \item \(K_4 \rightarrow^{t_{reset},e_3} K_5\) where \(t_{reset}\) is defined in Fig~\ref{fig:p4} and the mode is again fixed thanks to the places \(p_{x_i}\) that forces \(x_i\) to select the tuple \(m_{e(x_i)}^2\) and \(x_0\) selects \(\fmset{p_{t}^{rename}}\). Firing \(t_{reset}\) results in the configuration \[K_5 = \xi(M) + \fmset{p_{start}} + \fmset{} + \fmset{\xi(m'_{e(x_1)}),\xi(m'_{e(x_2)}),\xi(m'_{e(x_3)})} = \xi(K) + \fmset{}\] since \(\xi(m'_{e(x_i)}) = m^2_{e(x_i)}\).
\end{itemize}
The other direction is symmetric.

\end{proof}

\end{document}